\documentclass[journal]{IEEEtran}
\usepackage{amsmath,amsfonts,amssymb,amsbsy,bm,paralist,theorem,color}
\usepackage{graphicx,algorithmic,algorithm,relsize}
\usepackage{multicol}
\usepackage{multirow}
\usepackage{subfigure}
\usepackage{cite}
\newtheorem{lem}{Lemma}

\newtheorem{prop}{Proposition}

\graphicspath{{fig/}}
\date{}

\definecolor{orange}{RGB}{255,107,0}
\definecolor{green}{RGB}{0,160,20}

\begin{document}
\title{Optimal Resource Allocation for Delay Minimization in NOMA-MEC Networks}
\author{Fang Fang,~\IEEEmembership{Member,~IEEE}, Yanqing Xu,~\IEEEmembership{Member,~IEEE},
	Zhiguo Ding,~\IEEEmembership{Fellow,~IEEE},\\
	Chao Shen,~\IEEEmembership{Member,~IEEE}, Mugen Peng,~\IEEEmembership{Fellow,~IEEE},
	and George K. Karagiannidis,~\IEEEmembership{Fellow,~IEEE}
	\thanks{The work of M. Peng was supported by the National Natural Science Foundation of China under No. 61921003 and 61831002, and the Beijing Natural Science Foundation under No. JQ18016. The work of C. Shen was supported in part by the NSFC, China, under Grant 61871027 and Grant U1834210, in part by the State Key Laboratory of Rail Traffic Control and Safety under Grant RCS2019ZZ002. This paper was presented in part at the IEEE Global Communications Conference, Waikoloa, HI, USA, 2019.}
	\thanks{Fang Fang is with the Department of Engineering, Durham University, Durham DH1 3LE, U.K. (e-mail: fang.fang@durham.ac.uk).}
	\thanks{Z. Ding is with School of Electrical and Electronic Engineering, The University of Manchester, M13 9PL, U.K. (e-mail: zhiguo.ding@manchester.ac.uk).}
	\thanks{Y. Xu and C. Shen are with the State Key Laboratory of Rail Traffic Control and Safety, Beijing Jiaotong University, Beijing 100044, China (e-mail: yanqing\_xu@bjtu.edu.cn, chaoshen@bjtu.edu.cn).}
	\thanks{M. Peng is with the State Key Laboratory of Networking and Switching Technology, Beijing University of Posts and Telecommunications, Beijing 100876, China (e-mail: pmg@bupt.edu.cn).}
	\thanks{G. K. Karagiannidis is with the Wireless Communications \& Information Processing (WCIP) Group, Department of Electrical and Computer Engineering, Aristotle University of Thessaloniki, 54124 Thessaloniki, Greece (e-mail: geokarag@auth.gr).}
}\maketitle

\begin{abstract}
Multi-access edge computing (MEC) can enhance the computing capability of mobile devices, while non-orthogonal multiple access (NOMA) can provide high data rates. Combining these two strategies can effectively benefit the network with spectrum and energy efficiency. In this paper, we investigate the task delay minimization in multi-user NOMA-MEC networks, where multiple users can offload their tasks simultaneously through the same frequency band. We adopt the \emph{partial offloading} policy, in which each user can partition its computation task into \textit{offloading} and \textit{locally} computing parts. We aim to minimize the task delay among users by optimizing their tasks partition ratios and offloading transmit power. The delay minimization problem is first formulated, and it is shown that it is a nonconvex one. By carefully investigating its structure, we transform the original problem into an equivalent quasi-convex. In this way, a bisection search iterative algorithm is proposed in order to achieve the minimum task delay. To reduce the complexity of the proposed algorithm and evaluate its optimality, we further derive closed-form expressions for the optimal task partition ratio and offloading power for the case of two-user NOMA-MEC networks. Simulations demonstrate the convergence and optimality of the proposed algorithm and the effectiveness of the closed-form analysis.
\end{abstract}

\vspace{-0.5 cm}
\section{Introduction}
Driven by the explosive emergence of new compute-intensive applications in the Internet of things (IoT), especially ultra-low-latency applications, such as virtual reality (VR) and augmented reality (AR), multi-access edge computing (MEC) was proposed to enhance the computing capability of the mobile devices \cite{PMMECSurvey2017}. However, energy consumption and delay reduction still remain critical issues in practical systems. Due to high spectral efficiency, non-orthogonal multiple access (NOMA) has been proposed to combine with MEC to simultaneously support multiple users and lower the transmission latency and energy consumption. The combination of NOMA and MEC can provide various benefits including massive connectivity, low latency and high energy efficiency and flexibility to combine with other technologies, e.g., massive multiple-input multiple-output (MIMO) and millimeter wave (mmWave) \cite{AKiani2018JIOT}. As such, the resource allocation design is still imperative for realizing a low-cost and low-latency wireless network. In this work, we mainly focus on the offloading delay minimization in NOMA-MEC networks, where NOMA uplink transmission and MEC with \emph{partial} offloading (i.e., the computation task of each user can be partitioned into two parts, one of which will be offloaded for remote computation, while the other will be computed locally by the mobile devices.) are adopted. Therefore, the joint optimization of transmit power and task partition will be challenging to achieve the minimum task delay\footnote{Note that the task delay includes the transmission delay (offloading delay) and computation delay. Other network delays such as queuing delay will be considered in future work.}.
\subsection{Literature}
In this section, we first provide an overview of the MEC architecture and the NOMA-MEC systems. The recent research on resource allocation for minimizing energy consumption and task delay is presented for NOMA-MEC systems. MEC has been regarded as a key technology for the next generation of wireless networks due to its superior performance on energy consumption reduction and latency reduction \cite{YMaoMECSureveys2017}. As shown in Fig. 1, MEC can support various types of users such as IoT devices, mobile phones and self-driving cars at the edge of the network, where the base stations (BSs) equipped with MEC servers can provide cloud-like computing services for mobile devices with computation-intensive and delay-sensitive tasks \cite{WShiIEEEIoT2016}. In MEC networks, the majority of computing tasks can be offloaded to the MEC server at the BS for the remote computation. After the task computation at BS, the task results can be downloaded to mobile devices \cite{SAMECSurvey2014,HZhangFogIEEEWC2017}. 
In MEC offloading, tasks can be \emph{binary} offloaded (i.e., the computation task cannot be partitioned and must be either fully offloaded to the MEC server or computed locally) \cite{YMaoTWC2017,TXTranTVT2018,FWangTWC2018,SBIEEETWC2018} or \emph{partial} offloaded \cite{CYouTWC2017}. 

A very popular scheme is the power-domain NOMA, where the successive interference (SIC) technique is applied at the receiver, by which the user with large channel gain can remove the interference from the user with small channel gain \cite{DaiSurvey2018}. As a result, multiple users can transmit signals simultaneously with lower interference than orthogonal multiple access (OMA). Due to high spectral efficiency, the resource optimization in NOMA can achieve superior performance than OMA in terms of system sum rate and energy efficiency \cite{YSunIEEETCOM17,FangJSAC17}.

Driven by the superior performance of NOMA over OMA, NOMA uplink transmission and downlink transmission were proposed to be applied to MEC networks to support multiple users transmit signals simultaneously with lower interference \cite{ZDing2018TCOM}. Despite the benefits of NOMA-MEC, there still are some challenges including resource allocation, reliability and mobility, security and privacy, etc. Among them, resource optimization in NOMA MEC plays an important role in reducing task delay and energy consumption. Recently, some research has been carried to study the combination of NOMA and multi-user MEC networks  \cite{ZDing2018TCOM}. Most existing research works focus on computational resource allocation (e.g., tasks assignment) and communication resource allocation (e.g., offloading power and subchannel allocation)\cite{FWang2018TCOM,YPanCL2018,XCaoIoT2019,ZDing2018WCL,ZDing2018WCLDely,YWuJSTSP2019,LPSIC2018IoT,FangGlobecome2019}. How to allocate the network resources to reduce energy consumption and delay is an important issue to be addressed in NOMA-MEC. Specifically, in \cite{FWang2018TCOM}, the offloading tasks and offloading power levels of each user as well as the SIC decoding order were optimized to minimize the system energy consumption in NOMA-MEC. Moreover, in \cite{YPanCL2018}, the total energy consumption is minimized by optimizing the transmit powers, transmission time allocation, and task offloading partitions. The energy-efficient power allocation, time allocation and task assignment were proposed to minimize the energy consumption for MEC networks \cite{XCaoIoT2019}. To combine the advantages of OMA and NOMA systems \cite{FWang2018TCOM,YPanCL2018,ZDing2018TCOM}, a hybrid  NOMA-MEC scheme, where a user can first offload parts of its tasks within the time slot allocated to another user and then offload the remaining task during the time slot solely occupied by itself, was proposed in \cite{ZDing2018WCL,ZDing2018WCLDely}. In \cite{ZDing2018WCL}, to minimize the energy consumption, the transmit power and time slot were jointly optimized. Subsequently, the delay minimization was investigated via resource allocation in \cite{ZDing2018WCLDely}. Since the delay minimization is a crucial issue in NOMA-MEC, the resource optimization on offloading task assignment, transmit power and offloading time was investigated to reduce the task delay \cite{YWuJSTSP2019,LPSIC2018IoT,FangGlobecome2019}. Specifically, the overall delay of the computation tasks was minimized by a proposed algorithm of optimizing the offloading workload, offloading and downloading duration in \cite{YWuJSTSP2019}.
 Besides the computational resource, the SIC decoding order was optimized to reduce the task delay for NOMA enabled narrowband Internet of Things (NB-IoT) systems \cite{LPSIC2018IoT}. Different from the existing work \cite{YWuJSTSP2019,LPSIC2018IoT}, in this work, we mainly focus on the task delay minimization among all users in NOMA-MEC, in which the optimal solution is provided. A preliminary investigation on this problem was reported in \cite{FangGlobecome2019}\footnote{In this paper, we provide detailed proof for a key theorem, complexity analysis for the
proposed algorithms, and extensive simulation results. In addition, we also investigate a more
elaborate system model with the limited computing resource at the BS.}.
\begin{figure*}[t]
\centering
\graphicspath{{./figures/}}
\begin{minipage}[t]{0.48\textwidth}
\centering
\graphicspath{{./figures/}}
\includegraphics[width=0.95\linewidth]{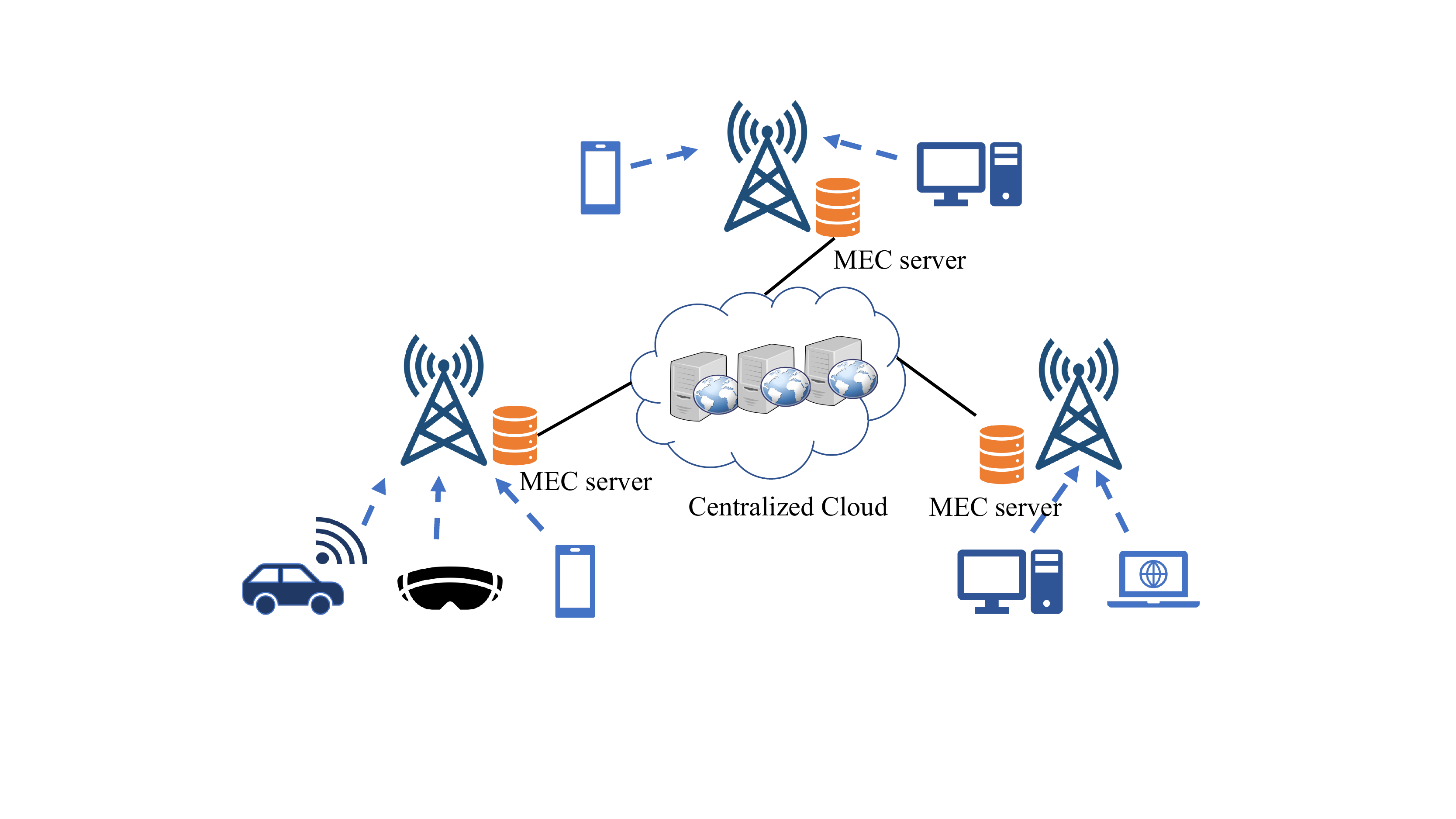}\\
	\caption{The architecture of MEC.} \label{Fig00}
\end{minipage}
\begin{minipage}[t]{0.48\textwidth}
\centering
	\graphicspath{{./figures/}}\includegraphics[width=0.95\linewidth]{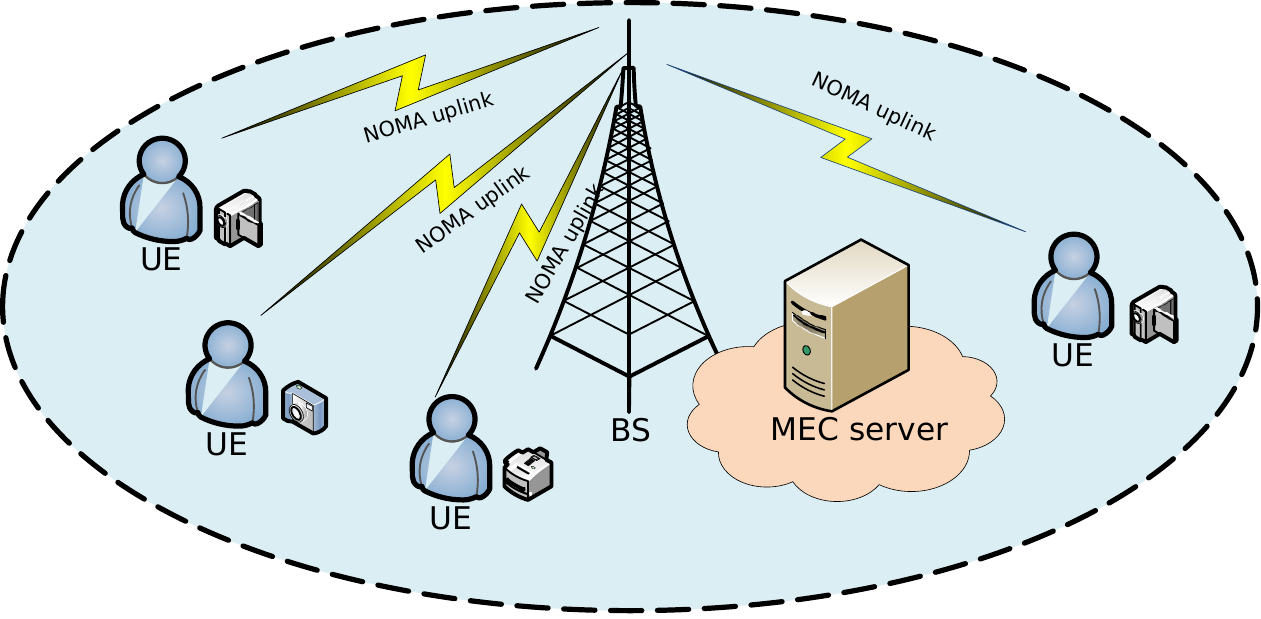}\\
	\caption{The uplink NOMA-MEC Network.} \label{Fig0}
\end{minipage}
\end{figure*}
\subsection{Motivation and Contribution}
Following a different approach from existing works, which mainly focus on energy minimization \cite{FWang2018TCOM,ZDing2018WCL,AKiani2018JIOT,SHanIoT2019,ZYEEGC2018,ZSongCL2018,QGEEDL2018GC} and fully offloading \cite{YMaoTWC2017,TXTranTVT2018,FWangTWC2018,SBIEEETWC2018},
in this paper, we focus on the task delay minimization by applying NOMA uplink transmission to the MEC \emph{partial offloading}\footnote{\emph{Partial offloading} is possible to implement for mobile applications that are composed of multiple procedures/components.}. There are the two main motivations for this work: First, from the communication perspective, the delay minimization is an important issue for MEC networks, especially for the offloading phase, which may cause a serious delay. Second, even though the delay minimization has been investigated in NOMA-MEC \cite{ZDing2018WCLDely,YWuJSTSP2019,LPSIC2018IoT}, insights of delay minimization for the NOMA-MEC network have not been obtained. Unlike the hybrid NOMA-MEC system in \cite{ZDing2018WCLDely}, we consider a more complex system, where pure NOMA-MEC and partial offloading are considered to minimize the task delay. The authors in \cite{YWuJSTSP2019,LPSIC2018IoT} have not considered the transmit power allocation, which is an important factor in reducing the task offloading delay. In this paper, we focus on \textit{data-partitioned-oriented applications}, in which the amount of data to be processed can be known beforehand, and the execution can be parallelized into processes. This requires that the data can be partitioned into subsets of any size. In practice, many mobile applications are composed of multiple procedures/components making it possible to implement partial computation offloading. However, only some partitions may be possible. That indicates that the optimal solution should be further quantized in practical implementation. Therefore, the results that we provide in this paper can be understood as a benchmark or an upper bound of the performance of any realistic offloading strategy [31]. The main contributions of this paper are listed as follow:
\begin{itemize}
\item [$\bullet$] In this paper, we apply NOMA into a multiuser MEC network where multiple users can offload their tasks to the MEC server simultaneously via the same frequency band. In this paper, we focus on \textit{data-partitioned-oriented applications}, in which the amount of data to be processed can be known beforehand, and the execution can be parallelized into processes. We assume that the data can be partitioned into subsets of any size, despite that, in practice, only some partitions may be possible. That indicates that the optimal solution should be further quantized in practical implementation. Therefore, the results that we provide in this paper can be understood as a benchmark or an upper bound of the performance of any realistic offloading strategy \cite{OM2015}. Considering the \emph{partial} task offloading scheme, the energy consumption limitation and offloading power limitation, the task completion time minimization problem is formulated as a nonconvex problem. Thus it is polynomial time unsolvable. By analyzing the properties of the formulated problem, some significant insights are revealed, and the corresponding propositions and Lemma are proposed to equivalently transform the original formulated problem into a simplified form. Based on those analytical results, the quasi-convexity of the transformed problem is proved. Therefore, a bisection searching  (BSS) based algorithm is proposed to find the globally optimal solution to the transformed problem. In the proposed algorithm, the original problem is solved by equivalently solving a series of feasibility subproblems. Moreover, we analyze the complexity of the proposed algorithm.  The convergence and optimality of the proposed BSS algorithm are evaluated by simulation results.  
\item [$\bullet$] Motivated by the practical applications, we focus on the two-user case to reduce the decoding complexity of SIC. To further reduce the complexity of the proposed BSS iterative algorithm, a corresponding proposition is proposed to equivalently transform the original problem into a convex problem. The convexity of the transformed problem is proved. Moreover, the closed-form expressions of the task partition ratios and offloading power are derived by exploiting the Lagrangian approach, i.e., Karush-Kuhn-Tucker (KKT) conditions, for the two-user NOMA MEC network. The simulation results demonstrate that the proposed BSS iterative algorithm matches our derived optimal solution, which reveals that our proposed BSS algorithm can converge to the optimal solution. Compared with existing schemes, i.e., fully NOMA system and fully OMA system, our proposed algorithm can provide the minimum task completion time in the simulation results. 
\item[$\bullet$] We also consider a scenario where the MEC server has limited computing resources. In this case, the computing time at the MEC server cannot be ignored. The optimal solution is derived to minimize the task delay for NOMA-MEC networks by considering computing time at the MEC server. A comparison of simulations is provided to demonstrate the superior performance of the proposed solution.
\end{itemize}

\subsection{Organization}

The organization of this paper is as follows. We introduce
the system model and formulate the task offloading time minimization problem in Section II. The BSS algorithm is proposed in Section III. In Section IV, the optimal closed-form solution is derived for the two-user case. Simulation results are presented in Section IV. Section V provides the optimal solution for limited computing resource MEC servers. Section VI concludes this paper.

\section{SYSTEM MODEL AND PROBLEM FORMULATION}
\subsection{NOMA-enabled Multi-user MEC Networks}
As depicted in Fig. \ref{Fig0}, the multi-user NOMA-MEC network consists of $M$ users randomly distributed in a single cell and one BS equipped with the MEC server located in the cell center. We assume that the BS and all users are equipped with a single antenna. The indices of users are defined as $m\in\{1,2,\cdots,M\}$. Denote the offloading task partial factor of User $m$ ($U_m$) by $\beta_m$, thus $(1-\beta_m)$ is the task partial factor of the task to locally compute at $U_m$, where $\beta_m \in (0,1)$. By implementing NOMA, the SIC technique is applied at the MEC server based BS. Define ${h}_{m}$ as the channel gain from $U_m$ to the MEC server. Without loss of generality, the channel gains of $M$ users are sorted as $ |g_{1}|\leq|g_{2}|\leq \cdots \leq |g_{M}|$. In this paper, we adopt block-fading channels, which indicates that the CSI will stay constant in a resource block but varies independently across different resource blocks. In this work, we assume that the BS knows the perfect CSI of all users. Note that the overhead will be costly if each requires to know the global CSI. However, in this work, the proposed algorithm is centralized to avoid the significant overhead between users. Specifically, the BS performs the proposed algorithm to make decisions for users; then the BS will broadcast the decisions to all the users by one pilot sequence. Therefore, the information exchange between users is not included in this work. Moreover, the BS only needs to know the CSI once within one block resource.

In the uplink NOMA system, the user with higher channel gain should be assigned with larger transmit power \cite{UplinkSIC2014}. Thus the SIC decoding order is assumed as the decreasing order of channel gains. It indicates that the MEC server first decodes the information transmitted by $U_M$ and then the information of $U_{M-1}$, until $U_1$. Define $p_{m}$ as the offloading power of $U_m$, then the signal-to-interference-plus-noise-ratio (SINR) of $U_m$ received at the MEC server can be written as
\begin{equation}
\Gamma_{m}^{off}=\frac{|h_{m}|^2p_{m}}{\sum \limits_{j=1}^{m-1}|h_{j}|^2p_{j}+1}
\end{equation}
where $h_m=\frac{g_m}{\sigma^2}$ is the channel gain normalized by $\sigma^2$, and $\sigma^2$ represents zero-mean complex additive white Gaussian noise (AWGN) power. Denote the system bandwidth by $B$ Hz, then the achievable data rate of $U_m$ can be written as
\begin{equation}\label{R_m_n_off}
\begin{aligned}
R_{m}
=B\log_2\left(\frac{\sum \limits_{i=1}^{m}|h_{i}|^2p_{i}+1}{\sum \limits_{j=1}^{m-1}|h_{j}|^2p_{j}+1}\right).
\end{aligned}
\end{equation}
Therefore, applying uplink NOMA transmission, the total offloading data rate can be written as
\begin{equation}
	R=\sum\limits_{m=1}^MR_m=B\log_2\left( \sum\limits_{i=1}^m|h_{i}|^2 p_{i} +1\right).
\end{equation} 
According to the MEC computation model, each user is enabled to offload a part of its task to the MEC server for remote computation. After computation, the task result can be downloaded from the MEC server to users. The task of $U_m$ can be described by two parameters $(L_m, C_m)$ where $L_m$ is the input number of bits for this task, and $C_m$ denotes the number of CPU cycles required to compute one bit of this task. In this paper, we assume that the downloading time from the BS equipped with the MEC server to the users is negligible. Two reasons are supporting this assumption. First, the size of task results are generally small \cite{CWangTVT2017,ZTTVT2018,XLyuTCOM2018,YHeTVT2018} and the BS generally has the comparatively more power to transmit the task result than the user offloading tasks to the BS. As a result, the downloading time is much shorter than the offloading time. Second, the downloading time optimization problem is more related to resource allocation from the perspective of BS including transmit power allocation and subchannel allocation. However, in this work, we mainly focus on resource allocation from the perspective of users including offloading power and task partition.

\subsubsection{Time Consumption for Offloading}
In this phase, each user will offload part of its task to the MEC server for remote executions. According to the achievable offloading data rate of $U_m$ \eqref{R_m_n_off}, the task offloading time from $U_m$ to the MEC server can be written as 
 \begin{equation}\label{T_m_off} 
  \begin{aligned}
 T_{m}^{off}=\frac{\beta_{m}L_m}{R_{m}}.
  \end{aligned}
  \end{equation}
The offloading energy consumption for $U_m$ is
 \begin{equation}\label{E_m_off}
\begin{aligned}
E_m^{off}=T_{m}^{off}p_{m}.
\end{aligned}
\end{equation}
\subsubsection{Mobile Execution Time}
For $U_m$, partial task $\beta_mL_m$ is offloaded to the MEC server for remote computation, and the remaining task $(1-\beta_m)L_m$ is computed locally. Denote the CPU frequency at $U_m$ is $f_m^{loc}$ (in cycles per second), then the local computation time of $U_m$ is given by
  \begin{equation}\label{T_m_loc}
 \begin{aligned}
 T_{m}^{loc}=\frac{(1-\beta_{m})L_mC_m}{f_m^{loc}}.
 \end{aligned}
 \end{equation}
The power consumption of the CPU cores at the MEC server is $P_m=\kappa_m(f_m^{loc})^3$, where $\kappa_m$ denotes the effective capacitance coefficient for each CPU cycle of the local user $U_m$ \cite{YMaoTWC2017}. Thus the energy consumption of computing at the mobile device $U_m$ can be written as
\begin{equation}\label{E_m_comp}
\begin{aligned}
E_{m}^{c}=p_{m}^{CPU}T_{m}^{loc}=\kappa_m(1- \beta_m) L_mC_m (f_m^{loc})^2.
\end{aligned}
\end{equation}

\subsection{Problem Formulation}
We first consider the scenario, where the MEC server at BS has considerable computation resource, which indicates the computation time at the BS can be ignored compared to the offloading time. For each user, a partial task will be offloaded to the MEC server for remote computation, and the remaining task will be computed locally. Each user's task will cost time to execute. Therefore, in this paper, we aim to minimize the maximum task completion time among the users, by optimizing the task assignment (i.e., offloading task ratio $\beta_m$) and communication resources (i.e., the offloading transmit power $p_m$). The computing time at the MEC server can be ignored compared to the offloading time, due to the high computing capacity of the MEC server. 
By considering the local computing time, the task completion time of $U_m$ can be written as
\begin{equation}\label{T_m}
T_{m}= \max \left\{T_{m}^{off},\ T_m^{loc}\right\}.
\end{equation}
Therefore, the task completion time minimization problem can be formulated as
\begin{subequations}\label{Prob:T_min}
\begin{align}
\mathop {\min\ \max }\limits_{{\{ \bm{\beta},\bm{p}\}}}  &\left\{T_{m}^{off},T_m^{loc},\forall m\right\} \\
 	 \text{s.t.}  \quad \quad \quad &0 \leq \beta_{m}\leq 1, \forall m,\label{eq:beta_con}\\   
    &0\leq p_{m} \leq P_{\max}, \forall m,\label{eq:p_con}\\ 
 	&E_m^c+E_m^{off}\leq E_{\max}, \forall m,\label{eq:E_m_con} 
\end{align}
\end{subequations}
where $\bm{\beta}=[\beta_1,\cdots,\beta_M]^T$ and $\bm{p}=[p_1,\cdots,p_M]^T$. Constraint \eqref{eq:beta_con} specifies the range of the offloading task ratio, while constraint \eqref{eq:p_con} describes the range of the offloading transmit power. Furthermore, constraint \eqref{eq:E_m_con} guarantees that the energy consumed at each mobile user is limited to the maximum energy consumption $E_{\max}$. Since problem \eqref{Prob:T_min} is nonconvex, next, we propose a transformation and an optimal solution.

\section{Optimal Solution for the Multi-user case}
 \subsection{Significant Observations and Problem transformation}
  In this section, we focus on the multi-user case and attain the minimum task delay for multi-user NOMA-MEC networks. A BSS iterative algorithm is proposed to find the global solution to problem \eqref{Prob:T_min}. Before solving this problem, we propose the following Proposition and Lemma.
\begin{prop} \label{T1=T2}
 In order to minimize the maximum task completion time of different users, i.e.,
\begin{equation}\label{T_m^off}
 	\mathop {\min\ \max }\limits_{\{\beta_m, \beta_{m'}, p_m, p_{m'}\}}\left\{T_{m}^{off},\ T_{m'}^{off}\right\},
 	\end{equation}
 	the task offloading ratios and transmit powers $\{\beta_m,p_m\}$ will be optimized, to make their offloading time equals to each other's, i.e., $T_m^{off}=T_{m'}^{off}$.
 	\end{prop}
 \begin{proof}
 	The proof is provided in Appendix \ref{T_1_off=T_m_off}.
 \end{proof}
 	
 	Based on Proposition \ref{T1=T2}, we adopt a pure NOMA offloading scheme, where multiple users can offload their tasks simultaneously in the same frequency band. Next, we provide a Lemma, which equivalently transforms the optimization problem into a simplified one.

 \begin{lem}\label{Lemma1}
 	Let's consider a single-cell NOMA-MEC network, where $M$ users offload their signals to the BS, equipped with MEC server within the same transmit time $T$,
 	\begin{equation}\label{Tm}
     T=T_{m}^{off}=T_{m'}^{off},\ \forall m\neq m'. 
 	\end{equation}
 	Then, \eqref{Tm} can be equivalently transformed to 
 	\begin{equation}\label{Tm2}
 	\tilde{T}_m^{off}=\frac{\sum\limits_{i=1}^{m}\beta_{i}L_i}{B\log_2\left( \sum\limits_{i=1}^m|h_{i}|^2 p_{i} +1\right)}, \quad \forall m.  
 	\end{equation}
 \end{lem}
 \begin{proof}
 	The proof is given in Appendix \ref{Lemma1Proof}.
 \end{proof}
Based on Lemma \ref{Lemma1} and Proposition \ref{T1=T2}, we have $\tilde{T}_m^{off}=T_m^{off}=T_{m'}^{off}, \forall m \neq m'$. Thus, the task completion minimization problem for the multi-user NOMA-MEC can be reformulated as:
\begin{subequations}\label{Prob:T_min_2}
	\begin{align}
	\mathop {\min \quad \max }\limits_{ \{\bm{\beta},\bm{p}\} }  &\left\{\tilde{T}_m^{off},T_m^{loc},\forall m\right\} \\
	\text{s.t.}  \quad \quad \quad &\beta_{m}\in [0,1], \forall m, \label{eq:beta_con_2}\\   
	&0\leq p_{m}\leq P_{\max},\forall m,\label{eq:p_con_2}\\ 
	&E_m\leq E_{\max},\forall m.\label{eq:E_1_con}
	\end{align}
\end{subequations}
where $E_m=  \kappa_m \beta_m L_m (f_m^{loc})^2+\tilde{T}_m^{off}p_m, \forall m$ is the energy consumed at $U_m$.
 This problem is  nonconvex due to the nonconvexity of $\tilde{T}_m^{off}$ with respect to $\{\beta_m,p_m\}$. Thus it is challenging to obtain its global optimum within polynomial time. 

\subsection{BSS Iterative Algorithm}
\begin{algorithm}[!t] \small
	\caption{~BSS algorithm for problem \eqref{Prob:T_transfer} }\label{Alg1}
	\begin{algorithmic}[1] 
		\STATE {{\bf Initialization:} Set $\alpha_T^{\min}=0$, $\alpha_T^{\max}=\max\{\frac{L_1C_1}{f_1^{loc}},\cdots,\frac{L_MC_M}{f_M^{loc}}\}$ and the accuracy $\epsilon=10^{-4}$.}
		\WHILE {$\alpha_T^{\max}-\alpha_T^{\min}> \epsilon$}
		\STATE {Set $\alpha_T=(\alpha_T^{\min}+\alpha_T^{\max})/2$.  }
		\STATE {Solve the convex feasibility problem \eqref{Prob:T_feasible} and find $C_{\alpha_T}$. }
		\IF{$C_{\alpha_T}\neq \phi$}
		\STATE Update $\alpha_T^{\max}=\alpha_T$.
		\ELSE
		\STATE Update $\alpha_T^{\min}=\alpha_T$.
		\ENDIF
		\ENDWHILE
		\STATE {{\bf Output:} $\alpha_T^*=\frac{\alpha_T^{\min}+\alpha_T^{\max}}{2}$, $\bm{\beta}^*$ and $\bm{p}^*$.}
	\end{algorithmic}\label{Al1}
\end{algorithm}

\begin{prop}\label{quasi}
	The objective function in problem \eqref{Prob:T_min_2} is strictly quasi-convex, given by the expressions \eqref{T_m_loc} and \eqref{T_m^off}. 
\end{prop}
\begin{proof}
	The proof is shown in Appendix \ref{QuasiProof}.
\end{proof} 

To obtain the optimal solution, problem \eqref{Prob:T_min_2} can be equivalently transformed to the following \cite{BoydConv2004}:
\begin{subequations}\label{Prob:T_transfer}
	\begin{align}
	\mathop {\min}\limits_{\{\bm{\beta},\bm{p},\alpha_T \}} \quad& \alpha_T \\
\text{s.t.}  \quad &\frac{\sum\limits_{i=1}^m\beta_{i} L_{i}}{B\log_2\left(1+ \sum\limits_{i=1}^m|h_{i}|^2 p_{i} \right)}\leq \alpha_T,\
\forall m \label{R_alphaT_M}\\
&\frac{(1-\beta_m)L_mC_m}{f_m^{loc}}\leq \alpha_T, \
\forall m \label{bm_alphaT} \\
& 0\leq\beta_{m}\leq 1,\  \forall m \label{b_0_1_M}\\
& 0\leq p_m\leq P_{\max},\  \forall m \label{p_0_max_M}\\
& \kappa_m(1- \beta_m) L_mC_m (f_m^{loc})^2+\alpha_Tp_m\leq E_{\max}, \ \forall m\label{EE_m}\\\nonumber
	\end{align}
\end{subequations}
where $\alpha_T$ is an auxiliary variable. Note that \eqref{Prob:T_transfer} is still nonconvex since the inequality constraint set is not convex in $\{\bm{\beta} ,\bm{p}, \alpha_T\}$. However, according to the quasi-convexity of \eqref{Prob:T_min_2}, the constraints \eqref{R_alphaT_M} and \eqref{EE_m} are convex when we fix $\alpha_T$. Therefore, \eqref{Prob:T_transfer} becomes a feasibility problem given by the fixed $\alpha_T$, and can be optimally solved by solving a series of convex feasibility subproblems. For a given $\alpha_T$, the feasibility problem can be formulated as:
\begin{subequations}\label{Prob:T_feasible}
	\begin{align}
	\text{ {find }}\quad&\{\bm{\beta} ,\bm{p} \} \\
	\text{s.t.} \quad&\eqref{R_alphaT_M}-\eqref{EE_m}.
	\end{align}
\end{subequations}
Note that the convex constraint set can be denoted by 
\begin{equation}
	C_{\alpha_T}=\{\{\bm{\beta} ,\bm{p} \}|\eqref{R_alphaT_M}-\eqref{EE_m}\}, \forall \alpha_T.
\end{equation}
In order to solve the problem in \eqref{Prob:T_transfer} through \eqref{Prob:T_feasible}, BSS can be utilized to find the optimal solution \cite{YXuTSP2017}. Specifically, Algorithm \ref{Al1} can be used to find the minimum task completion time of user tasks. In this algorithm, let's denote the optimal solution of the feasibility problem \eqref{Prob:T_transfer} by $\alpha_T^*$. We first initialize $\alpha_T$ by its lower and upper bounds. For a given $\alpha_T$, the problem \eqref{Prob:T_feasible} is feasible if $C_{\alpha_T}\neq\phi$, and holds that $\alpha_T\geq \alpha_T^*$. Then $\alpha_T$ will be decreased by step 6 of Algorithm \ref{Al1}. The problem \eqref{Prob:T_feasible} is infeasible if $C_{\alpha_T}=\phi$, and holds that  $\alpha_T<\alpha_T^*$. Then $\alpha_T$ will be increased by step 8 in Algorithm \ref{Al1}. The algorithm eventually converges to the unique optimal solution to the problem due to its strictly quasi-convexity \cite{ChongChiConvex}.
\subsection{Convergence Rate and Complexity Analysis}
Let us calculate the convergence rate of Algorithm 1. Denote the iteration index by $n$ and the optimal delay by $\alpha_T^{*}$. At  the $1$-st iteration of Algorithm 1, the bi-section range is $[\alpha_T^{\max},\alpha_T^{\min}]$. Let $\alpha_T^{(1)}=\frac{\alpha_T^{\max}-\alpha_T^{\min}}{2}$, we have $|\alpha_T^{(1)}-\alpha_T^{*}|\leq \frac{1}{2}\left(\alpha_T^{\max}-\alpha_T^{\min}\right)$. At the $2$-th iteration, the bi-section range is $\left[\alpha_T^{\max,(2)},\ \alpha_T^{\min,(2)}\right]$, $\alpha_T^{(1)}=\frac{\left(\alpha_T^{\max,(2)}-\alpha_T^{\min,(2)}\right)}{2}$. Note that one of ${\alpha_T^{\max,(2)},\ \alpha_T^{\min,(2)})}=\frac{1}{2}(\alpha_T^{\max}-\alpha_T^{\min})$ and the other one remain the same. Then we have 
$|\alpha_T^{(2)}-\alpha_T^{*}|\leq \frac{1}{2}(\alpha_T^{\max,(2)}-\alpha_T^{\min,(2)})=\frac{1}{2^2}(\alpha_T^{\max}-\alpha_T^{\min})$.
Therefore, by deduction, at the $n$-th iteration, we have $|\alpha_T^{(n)}-\alpha_T^{*}|\leq \frac{1}{2^n}(\alpha_T^{\max}-\alpha_T^{\min})$. Above all,  we can say that $\left\{\alpha_T^{(n)}\right\}_{n=1}^{\infty}$ converges to $\alpha_T^*$ with a rate of convergence ${\cal{O}}\left(\frac{1}{2^n}\right)$. The most important factor in the convergence rate is the number of iterations. In other words, the proposed algorithm will take more iterations when the convergence rate is low. Thus it is important to calculate the total number of iterations until the convergence.
Note that the computational complexity of Algorithm 1 comes from two aspects: one is from the iterations to find the optimal $\alpha_T^*$, and the other one is the complexity to solve the convex feasibility problem in each iteration for fixed $\alpha_T$. First, note that the searching range of the bi-section search approach determined by $\alpha_T^{\min}$ and $\alpha_T^{\max}$ will reduce by half after each iteration. Thus, for a given accuracy $\epsilon$, the total number of iterations of Algorithm 1 is $\log_2(\frac{\alpha_T^{\max}-\alpha_T^{\min}}{\epsilon})$. For the convex feasibility subproblem, an $\epsilon_2$-optimal solution can be obtained by the so-called ellipsoid method with $\log_2(\frac{1}{\epsilon_2})$ number of iteration. Therefore, the total computation complexity of Algorithm 1 is given by 
	\begin{equation}
		{\cal{O}}\left\{\log_2\left(\frac{\alpha_{\max} - \alpha_{\min}}{\epsilon}\right)\cdot \log_2\left(\frac{1}{\epsilon_2}\right)\right\}.
	\end{equation}
To demonstrate the feasibility of Algorithm \ref{Al1} in practical, we compare the complexity of the proposed algorithm with that of the uploading delay minimization (UDM) algorithm in \cite{YWuJSTSP2019}. Note that the complexity of UDM algorithm in \cite{YWuJSTSP2019} is ${\cal{O}}\left\{2I\times\log_2\left(\frac{x^{\text{upperbound}} - 1/T^{\text{up},\max}}{\epsilon}\right)\right\}$, where $I$ is the number of users and $x=\frac{1}{t^{\text{up}}}$ and where $t^{up}$ is the uploading time. In this work, if we set $\epsilon_2=0.001$, the complexity of Algorithm 1 is ${\cal{O}}\left\{10\times\left\{\log_2\left(\frac{\alpha_{\max} - \alpha_{\min}}{\epsilon}\right)\right\}\right\}$. Assume that the iteration for bi-section search are similar, i.e., $\log_2\left(\frac{\alpha_{\max} - \alpha_{\min}}{\epsilon}\right)\approx\log_2\left(\frac{x^{\text{upperbound}} - 1/T^{\text{up},\max}}{\epsilon}\right)$, the complexity of the proposed scheme in this work will be lower than the complexity of that of UDM algorithm in [21] when the user number is larger than five. Therefore, the proposed algorithm is more suitable for the network with large number of MEC users.

\section{Closed-Form Optimal Solution Derivation for the Two-user case} 
To further reduce the decoding complexity of SIC in the proposed BSS iterative algorithm, in this section, we focus on the two-user case and derive a closed-form solution, based on the insights and propositions obtained from problem  \eqref{Prob:T_min_2}. To further simplify the problem in \eqref{Prob:T_min_2}, the following proposition is provided.
 
\begin{prop} \label{T_loc=T_off}
In order to minimize the maximum task delay of each user, i.e.,
	\begin{equation}\label{T_m^loc}
	\mathop {\min\ \max }\limits_{\{\beta_m,p_m\}}\left\{T_{m}^{off},\ T_m^{loc}\right\}\  \forall m,
	\end{equation}
	 the offloading task ratio and powers will be optimized to the optimal solution $\{\beta_m^*,p_m^*\}$, which makes its offloading time equals to its local computing time, i.e., $T_m^{off}=T_m^{loc}$.
\end{prop}
\begin{proof}
	A Proof by contradiction is provided in Appendix \ref{T_m_loc=T_m_off}. 
\end{proof}

Based on Proposition \ref{T_loc=T_off}, we conclude that the optimal solution to problem \eqref{Prob:T_min_2} can be obtained when $T_m^{off}=T_m^{loc}$ for each user. Furthermore, according to Proposition \ref{T1=T2}, the optimal solution can be achieved only when the offloading time equals to each other. Considering the two-user case $|h_1|\leq|h_2|$, problem \eqref{Prob:T_min_2} can be written as
\begin{subequations}\label{Prob:T_min_2_2}
	\begin{align}
	\mathop {\min }\limits_{ \{\beta _{1},\beta _{2},p_{1}, p_{2} \}}  &\frac{\beta_1L_1+\beta_2L_2}{R(p_1,p_2)} \\
	\text{s.t.}  \quad  &\beta_{1}\in [0,1],\  \forall m=\{1,2\}, \\   
	&0\leq p_{1}\leq P_{\max}, \ \forall m=\{1,2\} \label{p_range}\\  
	&\kappa_m(1- \beta_m) L_mC_m (f_1^{loc})^2\\&+\frac{\beta_1L_1+\beta_2L_2}{R(p_1,p_2)}p_m \leq E_{\max},\  m=\{1,2\} \label{eq:E1}\\        
    &\frac{(1-\beta_m)L_mC_m}{f_1^{loc}}=\frac{\beta_1L_1+\beta_2L_2}{R(p_1,p_2)},\ m=\{1,2\}\label{eq1}\\
    &\frac{\beta_1L_1}{B\log_2(1+|h_1|^2p_1)}=\frac{\beta_1L_1+\beta_2L_2}{R(p_1,p_2)}\label{eq3}.
	\end{align}
\end{subequations}
where $R(p_1,p_2)=B\log_2(1+|h_1|^2p_1+|h_2|^2p_2)$. Note that the condition of Proposition \ref{T_loc=T_off} is that the user energy $E_{\max}$ is large enough to enable its offloading time to be equal to its local computing time. By applying this proposition into problem reformulation, the feasibility of the retransformed problem needs to be considered. 
If the problem is infeasible, the proposed BSS algorithm can be applied to solve the problem. In this case, each user's local computing time may not be equal to its offloading time. If the problem is feasible, the proposed optimal solution can be used to address the delay minimization problem.

 In problem \eqref{Prob:T_min_2_2}, the objective function is quasi-convex, as concluded from Proposition \ref{quasi}. However, constraint \eqref{eq:E1} is not convex set with respective to $\{\beta _{1},\beta _{2},p_{1}, p_{2}\}$. To simplify this problem, we first deal with equality constraints \eqref{eq1}-\eqref{eq3}. To solve the above problem and obtain the global optimum, we first equally transform this problem to an equivalent convex form via equality constraints. By using the equation \eqref{eq3}, we can replace the right sides of \eqref{eq1} with the left side of \eqref{eq3}. Then we have
 \begin{subequations}\label{equalities}
 		\begin{align}
	&(1-\beta_1)L_1C_1=\frac{\beta_1L_1}{B\log_2(1+|h_1|^2p_1)}f_1^{loc} \label{eq:beta11}\\
&(1-\beta_1)L_1C_1/f_1^{loc}=(1-\beta_2)L_2C_2/f_2^{loc}\\
&(1-\beta_1)L_1C_1/f_1^{loc}=\frac{\beta_1L_1+\beta_2L_2}{B\log_2(1+|h_1|^2p_1+|h_2|^2p_2)}. 
	\end{align}
\end{subequations}
After a series of calculations, 
problem \eqref{Prob:T_min_2_2} can be rewritten as 
\begin{subequations}\label{Prob:p1,p2}
	\begin{align}
	\mathop {\min }\limits_{ \{p_{1},p_{2}\}} \quad &\frac{L_1+L_2}{\frac{f_1^{loc}}{C_1}+\frac{f_2^{loc}}{C_2}+B\log_2(1+|h_1|^2p_1+|h_2|^2p_2)}\\
	\text{s.t.}  \quad&\eqref{p_range}-\eqref{eq:E1}.    
	\end{align}
\end{subequations}
\begin{prop} The problem \eqref{Prob:p1,p2} is convex.
\end{prop}
\begin{proof}
 The convexity proof is omitted due to the limited space.
\end{proof}
 According to equations \eqref{eq1}-\eqref{eq3}, once the optimal $p_1^*$ and $p_2^*$ are obtained, the optimal $\beta_1^*$ and $\beta_2^*$ can be calculated through:
\begin{subequations}
	\begin{align}
	&\beta_1^*=\frac{B\log_2(1+|h_1|^2p_1^*)}{\frac{f_1^{loc}}{C_1}+B\log_2(1+|h_1|^2p_1^*)}\\
	&\beta_2^*=1-\frac{(1-\beta_1^*)L_1C_1f_2^{loc}}{L_2C_2f_1^{loc}}\\ \nonumber
	\end{align}
\end{subequations}

In the following, we focus on deriving the optimal closed-form expressions for $p_1^*$ and $p_2^*$. Since the problem in \eqref{Prob:p1,p2} is convex and satisfies Slater's condition, the KKT conditions can be exploited to derive the optimal closed-form solution by the following four cases.

\emph{\underline{Case 1}}: When \begin{equation}
\left\{
\begin{aligned}
&  P_{1,w}(p_2^*=P_{\max}) \geq P_{\max} \quad  & \\
&  P_{2,w}(p_1^*=P_{\max}) \geq P_{\max} &
\end{aligned}
\right.
\end{equation}
Define
\begin{equation}
P_{1,w}(p_2^*)= -\frac{\mathcal{W}_0\left(-\frac{B_1\log(2)}{|h_1|^2}2^{\left(-\frac{B_1}{|h_1|^2}+A_1\right)}\right)}{B_1\log(2)}-\frac{1+|h_2|^2p_2^*}{|h_1|^2},\label{p1range_p1}
\end{equation} 
where
$A_1=\frac{\kappa_1a_1(f_1^{loc})^3}{E_{\max}B}-\frac{b_1}{B}$, $B_1=\frac{a_1}{E_{\max}B}$, $a_1=L_1+L_2$ , $b_1=\frac{f_1^{loc}}{C_1}+\frac{f_2^{loc}}{C_2}$ and where $\mathcal{W}_0(\cdot)$ is Lambert $\mathcal{W}$ function, which is a single value function. 

 \begin{equation}
 P_{2,w}(p_1^*)= -\frac{\mathcal{W}_0\left(-\frac{B_2\log(2)}{|h_2|^2}2^{\left(-\frac{B_2}{|h_2|^2}+A_2\right)}\right)}{B_2\log(2)}-\frac{1+|h_1|^2p_1^*}{|h_2|^2}\label{p2range_p1}
 \end{equation} 
 where $A_2=\frac{\kappa_2a_1(f_2^{loc})^3}{E_{\max}B}-\frac{b_1}{B}$ and $B_2=\frac{a_1}{E_{\max}B}$. 
Thus we have
\begin{equation}\label{Case1Opt}
\left\{
\begin{aligned}
&p_1^*= P_{\max} \quad  &\\
&p_2^*=P_{\max}. \quad   &
\end{aligned}
\right.
\end{equation}

\emph{\underline{Case 2}}: When \begin{equation}
\left\{
\begin{aligned}
&  P_{1,w}(p_2^*=P_{\max}) \geq P_{\max}\quad  & \\
& P_{2,w}(p_1^*=P_{\max}) \leq  P_{\max}, &
\end{aligned}
\right.
\end{equation}
thus we have
\begin{equation}\label{Case2Opt}
\left\{
\begin{aligned}
&p_1^*= P_{\max} \quad  &\\
&p_2^*=P_{2,w}(p_1^*=P_{\max}). \quad   &
\end{aligned}
\right.
\end{equation}

\emph{\underline{Case 3}}: When \begin{equation}
\left\{
\begin{aligned}
& P_{1,w}(p_2^*=P_{\max})\leq P_{\max} \quad  & \\
& P_{2,w}(p_1^*=P_{1,w}) \geq P_{\max}, &
\end{aligned}
\right.
\end{equation}
thus we have
\begin{equation}\label{Case3Opt}
\left\{
\begin{aligned}
&p_1^*= P_{\max} \quad  &\\
&p_2^*=P_{2,w}(p_1^*=P_{\max}). \quad   &
\end{aligned}
\right.
\end{equation}

\emph{\underline{Case 4}}: When \begin{equation}
\left\{
\begin{aligned}
& P_{1,w}(p_2^*)\leq P_{\max} \quad  & \\
& P_{2,w}(p_1^*)\leq P_{\max}, &
\end{aligned}
\right.
\end{equation}
holds that
\begin{equation}\label{Case4}
\left\{
\begin{aligned}
p_1^*=&\kappa_2(f_2^{loc})^3-\kappa_1(f_1^{loc})^3+p_2^*\\
p_2^*=&-\frac{\mathcal{W}_0\left(-\frac{B_2\log(2)}{(|h_1|^2+|h_2|^2)}2^{\left(-\frac{B_2}{(|h_1|^2+|h_2|^2)}+A_2\right)}\right)}{B_2\log(2)}\\
&-\frac{1+|h_1|^2\left(\kappa_2(f_2^{loc})^3-\kappa_1(f_1^{loc})^3\right)}{(|h_1|^2+|h_2|^2)}, 
\end{aligned}
\right.
\end{equation}
where  $A_2=\frac{\kappa_2a_1(f_2^{loc})^3}{E_{\max}B}-\frac{b_1}{B}$ and $B_2=\frac{a_1}{E_{\max}B}$. 
\begin{proof}
The proof is provided in Appendix \ref{OptimalDerivation2}.
\end{proof}
From the above solution, we can directly obtain the optimal solution for the two-user case NOMA-MEC networks based on the parameters of channel gains and computing capabilities of users. The complexity is significantly reduced compared with the proposed Algorithm \ref{Alg1}. Note that the closed-form solution is extremely challenging to obtain for the multi-user case in NOMA-MEC. However, the closed-form solution can be applied to the scenario with a large number of users. In this case, user clustering/grouping can be implemented via matching theory or game theory or machine learning. Specifically, multiple users can be grouped into multiple clusters, and each cluster contains two users, where the closed-form solution can be applied.


\section{Task Delay minimization for limited computing resources MEC servers}
In this section, we consider a more complex scenario, where the MEC server at the BS has limited computing resources. This indicates that the computation time for the same task will be longer than that of the scenario with the considerable computation resource. In this case, the computation time at the BS cannot be ignored. After offloading $\sum\limits_{m=1}^{m=M}\beta_mL_m$ tasks to the MEC server, then the computing time at MEC server is given by
  \begin{equation}\label{T_m_loc2}
 \begin{aligned}
 T_{S}^{c}=\frac{\sum\limits_{m=1}^{M}\beta_mL_mC_S}{f_{S}},
 \end{aligned}
 \end{equation}
 where $C_S$ denotes the numbers of CPU cycles required to compute 1 bit, and $f_S$ denotes the CPU frequency (in cycles per second). Note that $T_S^c$ is a function of $\beta_m$ and $f_S$, which will increases when the task ratio increases and decreases when $f_S$ increases. When we consider $T_S^c$ into delay minimization, it will affect the offloading scheme.
Thus the energy consumption of computation at the MEC server can be written as
\begin{equation}\label{E_m_comp}
\begin{aligned}
E_{S}^{c}=\kappa_S\sum\limits_{m=1}^{m=M}\beta_mL_m (f_S)^2,
\end{aligned}
\end{equation}
where $\kappa_S$ denotes the effective capacitance coefficient for each CPU cycle of the MEC server. Considering pure NOMA offloading MEC networks, the task delay minimization can be written as 
\begin{subequations}\label{Prob:T_min_0}
	\begin{align}
	\mathop {\min \quad \max }\limits_{ \{\bm{\beta},\bm{p}\} }  &\left\{\tilde{T}_m^{off}+T_{S}^c,T_m^{loc},\forall m\right\}  \\
	\text{s.t.}  \quad \quad \quad &\beta_{m}\in [0,1], \forall m, \label{eq0:beta_con_2}\\   
	&0\leq p_{m}\leq P_{\max},\forall m,\label{eq0:p_con_2}\\ 
	&E_m\leq E_{\max},\forall m,\label{eq0:E_1_con}
	\end{align}
	\end{subequations}
	where \begin{equation}
		\tilde{T}_m=\tilde{T}_m^{off}+T_{S}^c.
	\end{equation}
	The proposed BSS algorithm can be still applied for this scenario, since the quasi-convexity of Problem \eqref{Prob:T_min_0} can be proved by similar steps as in Appendix \ref{QuasiProof}.
	
	For the two-user case, the transformed objective function can be written as 
	\begin{subequations}\label{Prob:p1,p2,new}
	\begin{align}
	\mathop {\min }\limits_{ \{p_{1},p_{2}\}} \quad &\frac{L_1+L_2}{\frac{f_1^{loc}}{C_1}+\frac{f_2^{loc}}{C_2}+\frac{1}{1/R+C_S/f_S}} \label{Objnew1}\\ 
	\text{s.t.}  \quad&\eqref{p_range}-\eqref{eq:E1}.    
	\end{align}
\end{subequations}
Since the objective function \eqref{Objnew1} is monotonically decreasing with $R$, solving Problem \eqref{Prob:p1,p2,new} is equivalent to replace the objective with $R$. The derived solution $p_1^*,p_2^*$ can still be applied to this scenario with the minimum delay \begin{equation}
	T_{\min}=\frac{L_1+L_2}{\frac{f_1^{loc}}{C_1}+\frac{f_2^{loc}}{C_2}+\frac{1}{1/R^*+C_S/f_S}}
\end{equation}
where $R^*=R(p_1^*,p_2^*)$. The performance of this solution is demonstrated in Fig. 2.

\section{Simulation Results and Discussion}
In this section, the performance of the proposed resource allocation schemes is evaluated by simulations. The proposed BSS algorithm and the optimal analyzed solution are compared with three benchmark schemes: 1. the orthogonal frequency-division multiple access (OFDMA) based partial scheme, where users offload their tasks by the OFDMA transmission. 2. the NOMA based full offloading scheme, where users offload its entire task to the BS for the remote computation via NOMA transmission. 3. the fully local computing scheme, where users compute the entire task locally. In the simulations settings, all users are randomly distributed in a single cell with a radius of 500 m. The channel gain from the $m$-th user to the BS is denoted by $h_{m}=g_{m}{(1+d_{m}^\alpha)^{-\frac{1}{2}}}$, where $g_{m}$ is a Rayleigh fading channel coefficient and $d_m$ is the distance from this user to the BS. The path loss factor $\alpha$ is 3.76. The AWGN power is $\sigma^2=BN_0$ where the AWGN spectral density is $N_0= -174$ dBm/Hz. For the computational resource at each user, we set $C_m=10^3, \forall m$ cycles per bit and $f_m=10^9, \forall m$ cycles per second.
\begin{figure}[t]
\centering
\graphicspath{{./figures/}}
\includegraphics[width=0.95\linewidth]{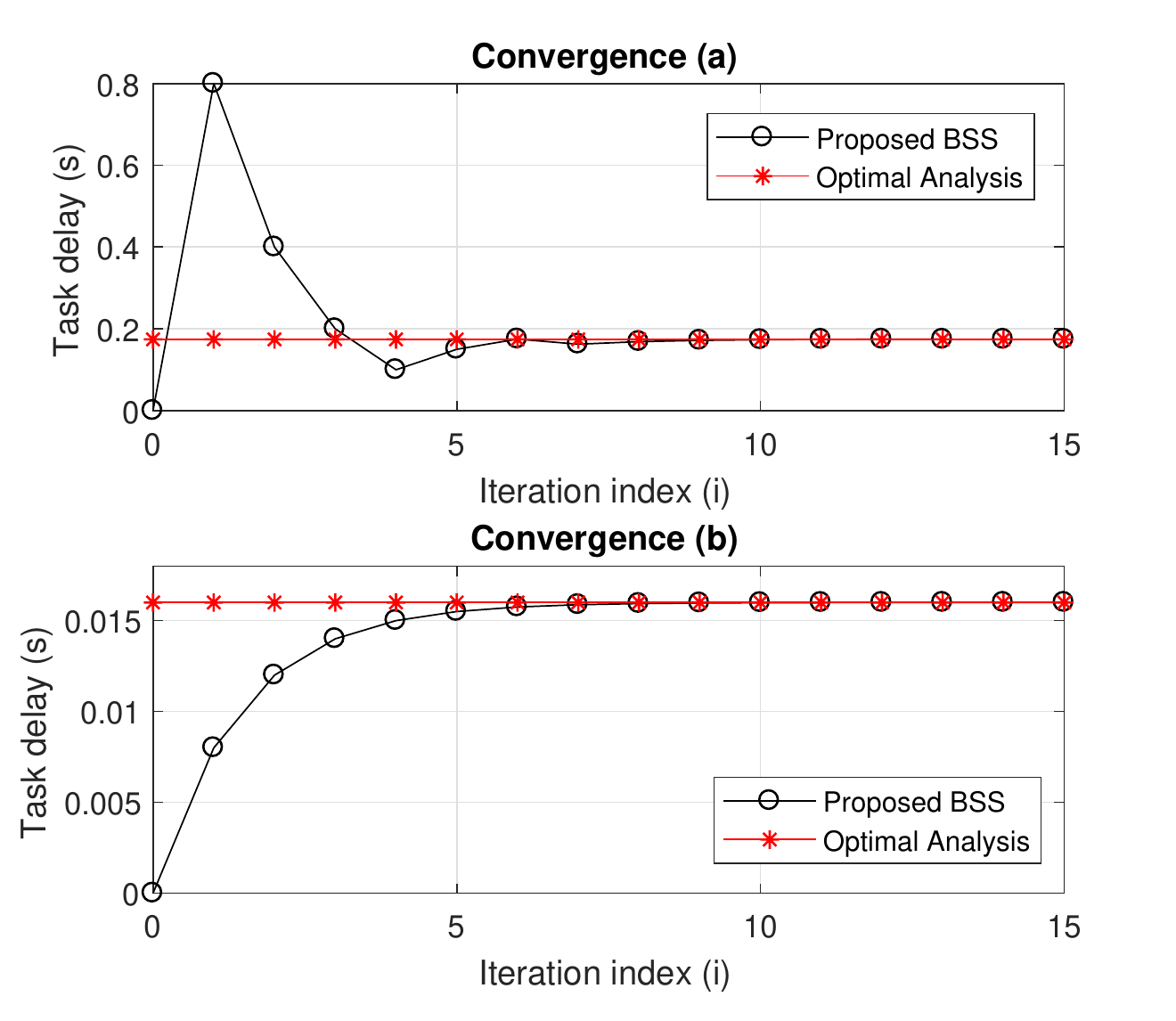}\\
	\caption{The convergence and optimality of Algorithm 1.} \label{Fig1}
\end{figure}
\begin{figure}[t]
\centering
	\graphicspath{{./figures/}}\includegraphics[width=0.95\linewidth]{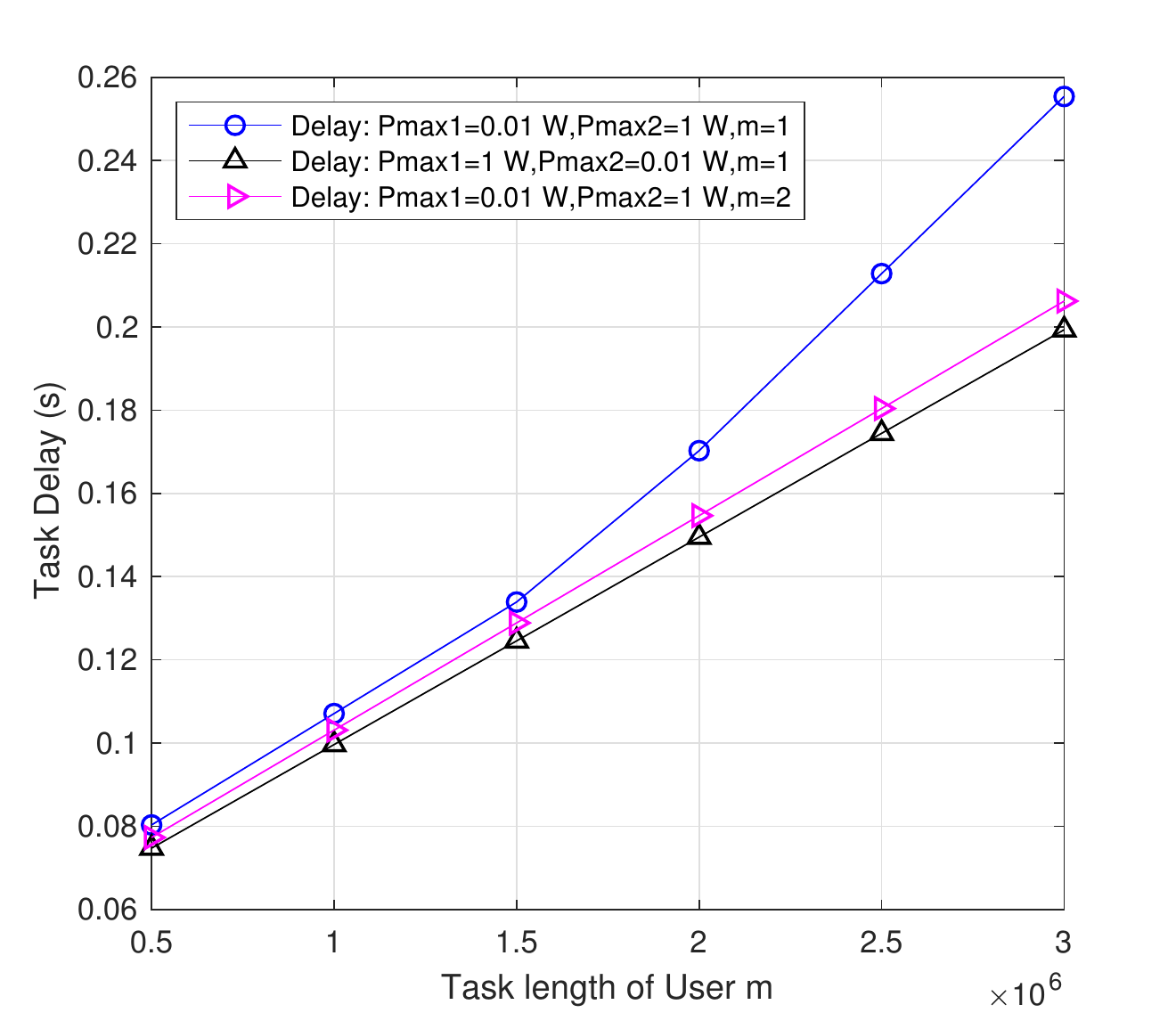}\\
	\caption{Delay comparison versus the task length of $U_m$.} \label{Fig55}
\end{figure}

 Fig. \ref{Fig1} shows the convergence and optimality of the proposed BSS algorithm. In Fig. \ref{Fig1}(a), we set the bandwidth to $B=1$ MHz. The length of computation tasks for each user is set to $L_1=L_2=1.6\times10^6$ bits. The maximum energy consumption for each user is $E_{\max}=0.2$ Joule, and the maximum power for each user is  $P_{\max}=0.01$ W, and the effective capacitance coefficient for each CUP cycle of the local users are $\kappa_m=10^{-28}\times[10,1]$. We initialize the lower bound on the delay to zero and the upper bound on the delay to the maximum fully local computing time. In Fig. \ref{Fig1}(a), there is a fluctuation before the convergence because the optimal delay is close to the lower bound. In Fig. 1(b), we set the bandwidth to $B=1$ Hz, and the length of computation tasks for each user is set to $L_1=L_2=1.6\times10^4$ bits. Since the offloading data rate is low due to the small bandwidth, most of the tasks will be computed by local users. In this case, the optimal delay is close to the upper bound. As a result, the delay at the first iteration is much lower than the optimal value (almost the upper bound), and it keeps increasing by each iteration until its convergence. From this figure, we can see that the proposed BSS algorithm converges within 10 iterations, which indicates the proposed scheme is practical. Moreover, the convergence point is perfectly matched with the optimal analytical solution, which indicates the proposed algorithm is optimal.

Fig. \ref{Fig55} illustrates the relationship between the task delay and the task size of each user with different $P_{\max}$. The parameter setting is the same as that in Fig. \ref{Fig1}. In Fig. \ref{Fig55}, the delay of each scheme decreases when the task length increases. From this figure, we can see that the scheme with $P_{\max}=1$ W can achieve a smaller delay than that with lower $P_{\max}=0.01$ W when the task length of $U_1$ increases. Moreover, the scheme with $P_{\max1}=1$ W can achieve a smaller delay than the scheme with $P_{\max2}=1$ W. This is because $U_1$ has a higher channel gain than $U_2$.

\begin{figure}[t]
\centering
\graphicspath{{./figures/}}
\includegraphics[width=0.95\linewidth]{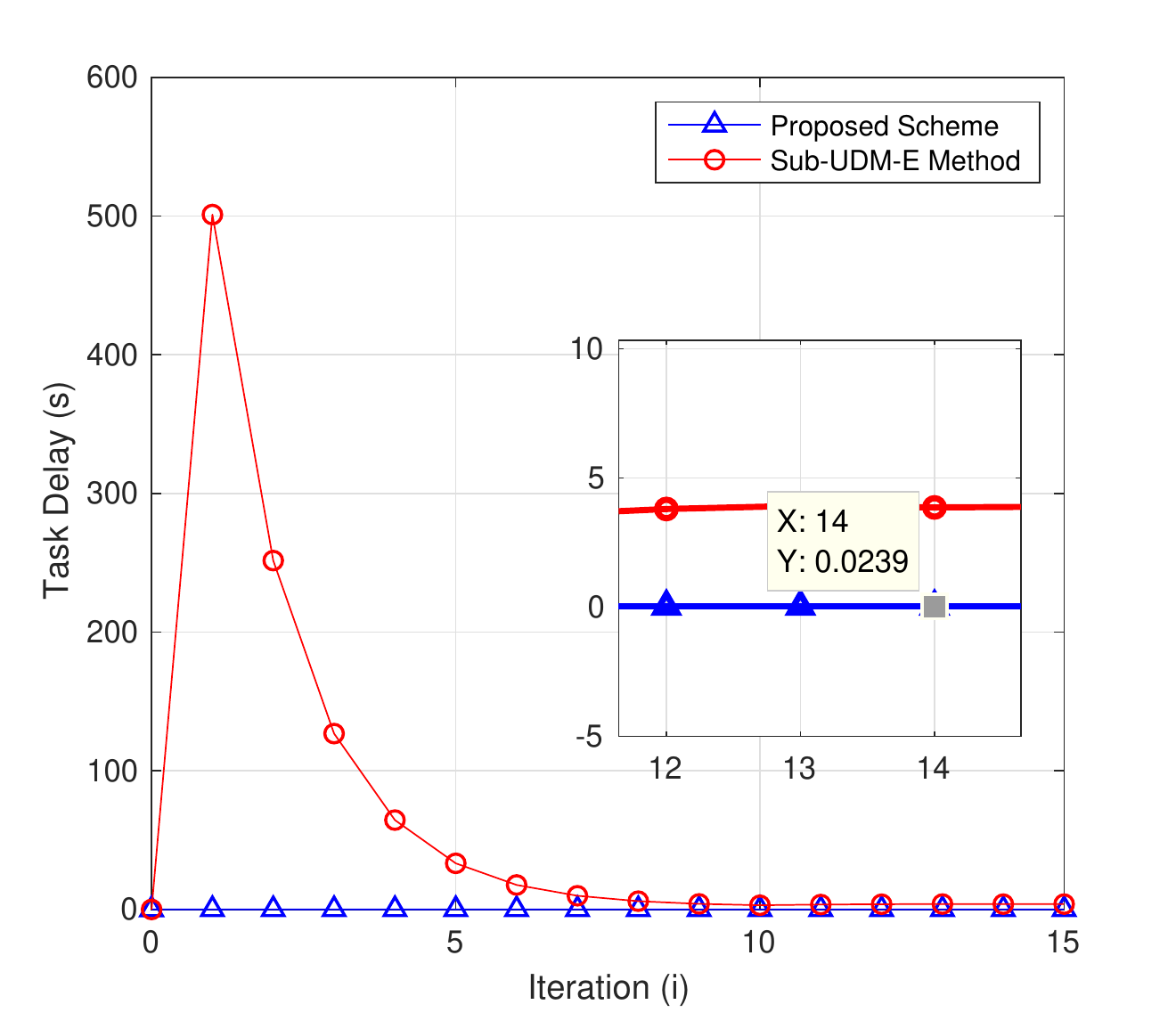}\\
	\caption{Delay comparison with the existing work \cite{YWuJSTSP2019}.} \label{Fig2}
\end{figure}
\begin{figure}[t]
\graphicspath{{./figures/}}
\centering
\includegraphics[width=0.95\linewidth]{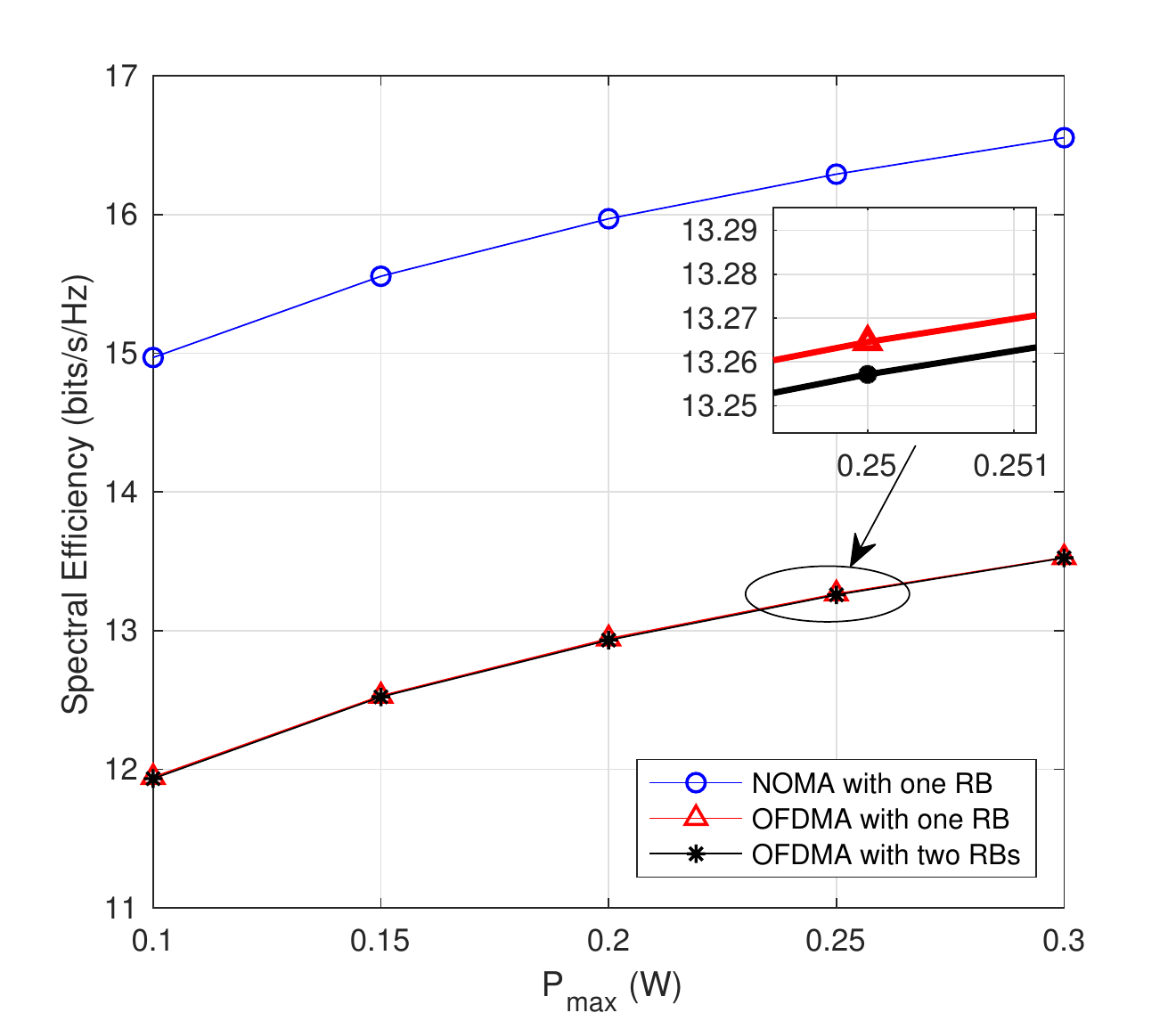}\\
	\caption{The offloading data rate comparison of NOMA based and OFDMA based schemes.} \label{FigSE}
\end{figure}

Fig. \ref{Fig2} presents the delay comparison of the proposed scheme with the UDM method in \cite{YWuJSTSP2019}. The length of computation tasks for each user is set to $L_1=L_2=1.6\times10^6$ bits. The maximum energy consumption for each user is $E_{\max}=0.2$ Joule, and the maximum power for each user is  $P_{\max}=0.01$ W, and the effective capacitance coefficients for each CUP cycle of the local users are $\kappa_m=10^{-28}\times[10,1]$. It can be observed that the proposed algorithm can achieve better performance on delay minimization. This is because the proposed method is optimal,  and Sub-UDM-E method in \cite{YWuJSTSP2019} only considered time optimization. In the proposed scheme, the task assignment and offloading power are optimized to minimize the task delay.
\begin{figure}[t]
\centering
\graphicspath{{./figures/}}
\includegraphics[width=0.95\linewidth]{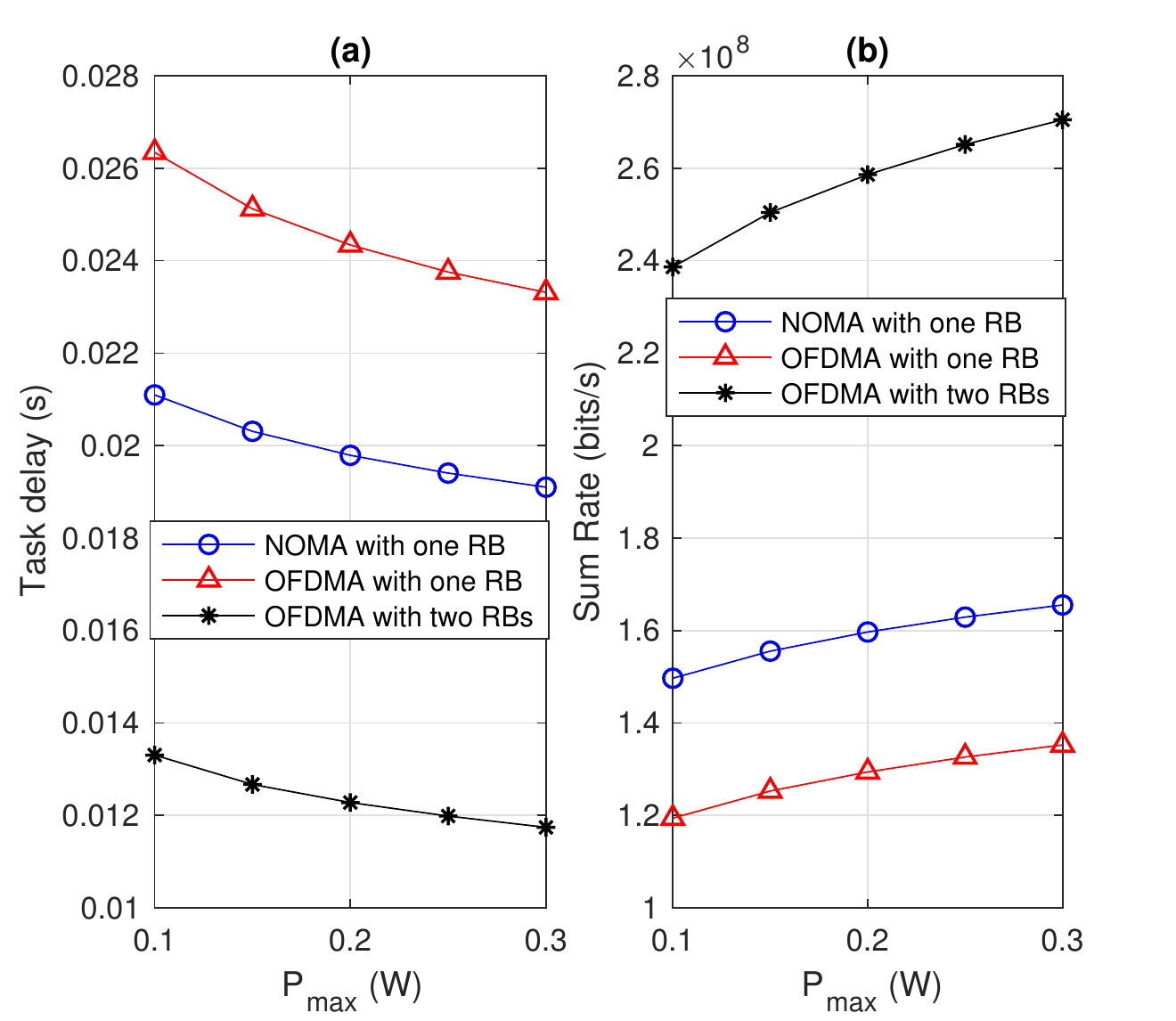}\\
	\caption{Task delay and sum rate comparison.} \label{FigTDSR}
\end{figure}
\begin{figure}
\centering	
\graphicspath{{./figures/}}
\includegraphics[width=0.95\linewidth]{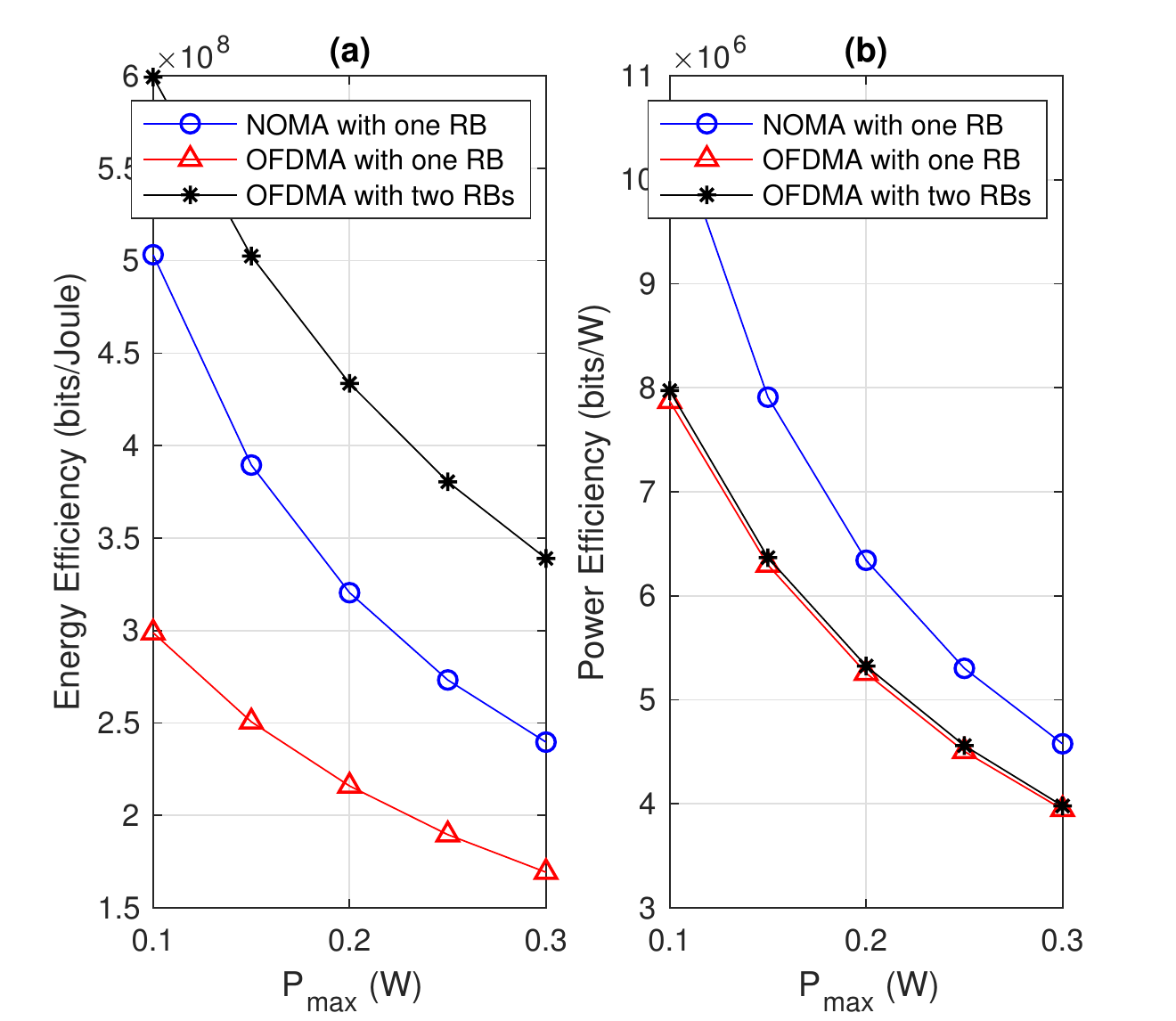}\\
	\caption{Energy efficiency and power efficiency comparison.} \label{FigEEPE}
\end{figure}
Fig. \ref{FigSE} illustrates the spectral efficiency comparison of NOMA offloading and OFDMA offloading schemes with different resource blocks (RBs) versus the maximum power. In this figure, we set $B=1$ MHz and $E_{\max}=0.2$ Joule. The tasks lengths of these two users are $L_1=L_2=1.6\times10^6$ bits. Considering the two-user case, we compare the spectral efficiency of the proposed delay minimization algorithm with three schemes, i.e., NOMA with one RB, OFDMA with one RB and OFDMA with two RBs. From this figure, we can see that NOMA with the proposed scheme yields higher spectral efficiency than the other two OFDMA schemes.

Fig. \ref{FigTDSR} shows the delay comparison and sum rate comparison of NOMA offloading and OFDMA offloading schemes versus the maximum power. The parameter setting is the same as that in Fig. \ref{FigSE}. From Fig. \ref{FigTDSR}, the OFDMA scheme with two RBs has the lowest delay and highest sum rate compared with the NOMA scheme with one RB and the OFDMA scheme with one RB. However, with one RB, the NOMA scheme always outperforms the OFDMA scheme in terms of task delay and sum rate.
 
In Fig. \ref{FigEEPE}, we provide the energy efficiency comparison and power efficiency comparison of NOMA offloading and OFDMA offloading schemes. The parameter setting is the same as that in Fig. \ref{FigSE}. In Fig. \ref{FigEEPE}(a), we set the constant power consumption of circuits as $Pc=0.1$ W. We can see that the energy efficiency of OFDMA with two RBs is better than that of NOMA with one RB. However, with the same RB, the energy efficiency of NOMA is better than OFDMA, which can be verified by \cite{FangIEEETrans16}. While in power efficiency (PE) comparison Fig. \ref{FigTDSR}(b), the NOMA scheme achieves the highest PE among these three schemes, and the OFDMA scheme with two RBs outperforms the scheme with one RB.

\begin{figure}[t]
\centering
\graphicspath{{./figures/}}
\includegraphics[width=0.95\linewidth]{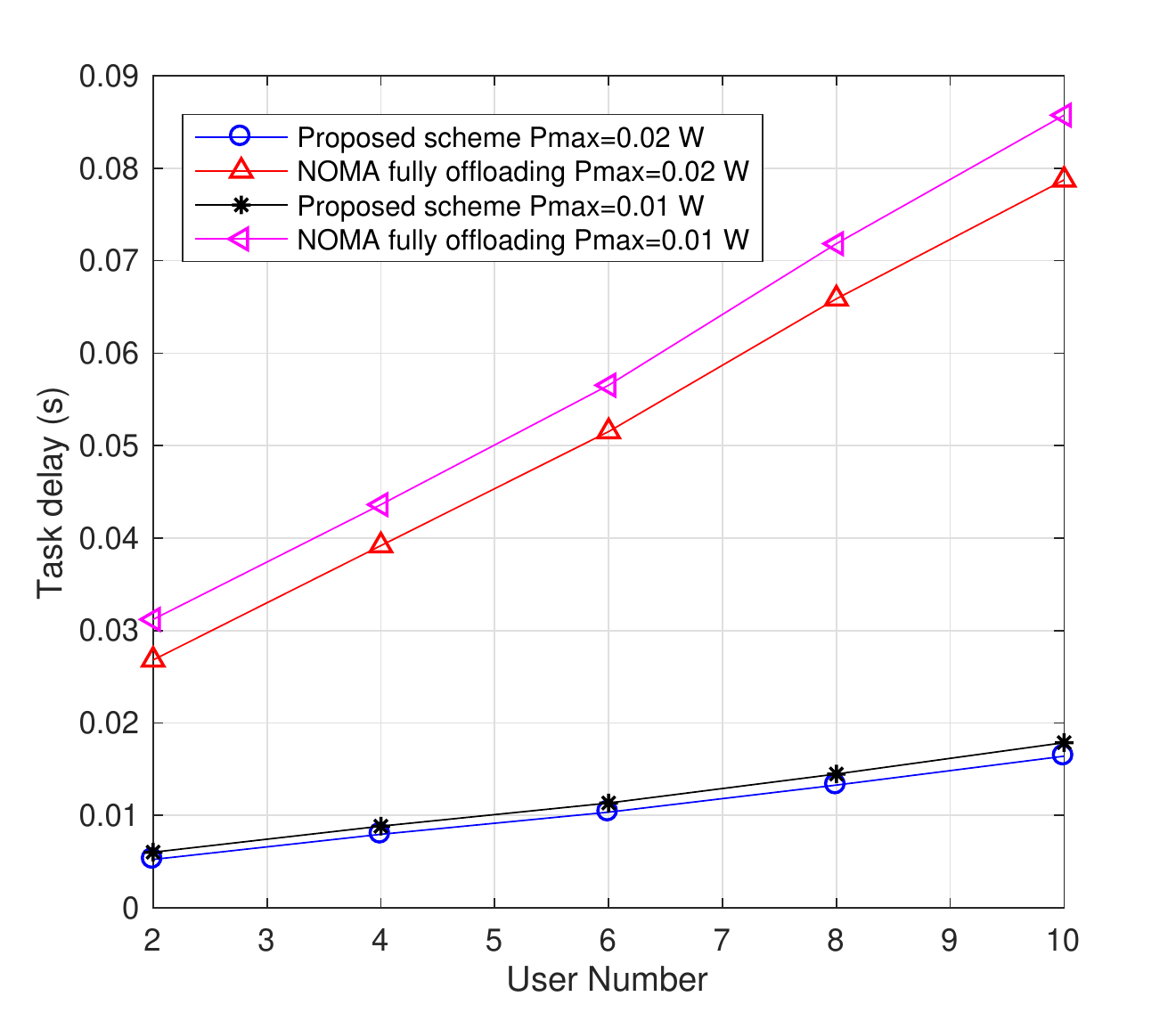}\\
	\caption{The task completion time versus user number with different powers.} \label{Fig7}
\end{figure}

Fig. \ref{Fig7} demonstrates the task completion time comparison between NOMA based partial offloading scheme and NOMA fully offloading scheme by considering different offloading power limits. In this figure, we set $B=1$ MHz and $L_m=1.6\times10^6$ bits, $E_{\max}=0.2$ Joule,  $P_{\max}=0.01$ W, and $\kappa_m=10^{-28}[10,1]$. In this figure, the task completion time increases when the user number increases. The scheme with a higher transmit power limit $P_{\max}=0.02$ W will cost less time than the scheme with $P_{\max}=0.01$ W. It also shows that NOMA based partial offloading scheme can provide better performance than the NOMA based fully offloading scheme.

\section{Conclusion}
In this paper, we have investigated the task delay minimization problem for NOMA enabled multi-user MEC networks. The task delay minimization problem has been formulated as a nonconvex one. By well utilizing the structure of the original problem, it has been equivalently transformed into a quasi-convex one. For the multi-user case, the BSS algorithm has been proposed to efficiently obtain the globally optimal solution. To further reduce the complexity and obtain more insights of the system, the delay minimization problem has been formulated for the two-user case. Closed-form optimal power allocation and offloading task ratio expressions have been derived for the two-user case based on analytic insights. Simulations have been provided to demonstrate the convergence and optimality of the proposed schemes, which can provide an effective solution to minimize task delay for a single-cell NOMA-MEC network. Furthermore, the solution of this work can also be applied to more complicated networks, e.g., multi-cell networks. In general, user grouping or user pairing via matching theory or game theory can be applied to decouple the multi-cell optimization problem into parallel subproblems, where the solution of this work can be applied. This study is out of the scope of this paper and will be treated as an important direction for future research.

\appendices
\section{Proof of Lemma \ref{Lemma1}} \label{Lemma1Proof}
We assume that $M$ users offload their tasks to the MEC server simultaneously through one subchannel. By applying the SIC technique at the receiver, the BS equipped with the MEC server decodes the signals in decreasing order of their channel gains,  
    $|h_{1}|^2\leq |h_{2}|^2 \leq \cdots \leq |h_{M}|^2$. 
We further assume that $M$ users transmit their tasks in the same period $T$, which indicates:
    \begin{equation}\label{T1=T2=T}
       T=T_{1}^{off}=T_{2}^{off}=\cdots=T_{M}^{off} 
    \end{equation}
   and
   \begin{equation}
       T=\frac{\beta_{1}L_1}{R_{1}} =\frac{\beta_{2}L_2}{R_{2}}=\cdots=\frac{\beta_{M}L_M}{R_{M}}.
    \end{equation}
    Since $\frac{a_1}{b_1}=\frac{a_2}{b_2}=\cdots=\frac{a_M}{b_M}$ can be written as $\frac{a_1}{b_1}=\frac{e_2a_1}{e_2b_1}=\cdots=\frac{e_Ma_1}{e_Mb_1}$ where $e_m=\frac{a_m}{a_1}=\frac{b_m}{b_1}$. Thus $\frac{a_1}{b_1}=\frac{a_1(1+e_2+e_3+\cdots+e_M)}{b_1(1+e_2+e_3+\cdots+e_M)}=\frac{a_1+a_2+\cdots+a_M}{b_1+b_2+\cdots+b_M}$. Therefore,
     \begin{equation}
       T=\frac{\beta_{1}L_1+\beta_{2}L_2+\cdots+\beta_{m}L_m}{R_{1}+R_{2,n}+\cdots+R_{m}} =\frac{\sum\limits_{i=1}^m\beta_{i}L_i}{\sum\limits_{i=1}^mR_{m}}, \forall m.
    \end{equation} 
    The offloading sum rate can be derived as 
    \begin{equation}
    \begin{aligned}
    	 \sum\limits_{i=1}^mR_{i}=&B\log_2(1+|h_{1}|^2p_{1})+B\log_2\left(\frac{1+|h_{1}|^2p_{1}+|h_{2}|^2p_{2}}{1+|h_{1}|^2p_{1}}\right) \\
    	 &+\cdots+B\log_2\left(\frac{1+\sum\limits_{i=1}^m|h_{i}|^2p_{i}}{1+\sum\limits_{i=1}^{m-1}|h_{i}|^2p_{i}}\right)\\
    	&=B\log_2\left(1+\sum\limits_{i=1}^m|h_{i}|^2p_{i}\right).\\
    \end{aligned}
    \end{equation}
    Define the number of transmitted bits as $L^{off}=\sum\limits_{m=1}^M \beta_{m}L_m$, and $R^{off}=B\log_2\left(1+\sum\limits_{m=1}^M|h_{m}|^2p_{m}\right)$. Then the offloading time can be written as
    \begin{equation}\label{Tm_off}
       T=\frac{\sum\limits_{i=1}^m\beta_{i}L_i}{B\log_2\left(1+\sum\limits_{i=1}^m|h_{i}|^2p_{i}\right)}, \forall m.
    \end{equation} 
    Now let's prove from \eqref{Tm_off} to \eqref{T1=T2=T}.
    When $m=1$, we have $T=T_1^{off}$.
    When $m=2$, we have $T=\frac{\sum\limits_{i=1}^2\beta_{i}L_i}{\sum\limits_{i=1}^2R_{2}}$.
    Since $\frac{\beta_1L_1}{R_1}=\frac{\beta_1L_1+\beta_2L_2}{R_1+R_2}$, we can have  $\frac{\beta_1L_1}{R_1}=\frac{\beta_2L_2}{R_2}$.
    By deduction, we can have \eqref{T1=T2=T}.
    We finish the proof of Lemma \ref{Lemma1}.
    
    \section{Proof of Proposition \ref{T1=T2}}\label{T_1_off=T_m_off}
    According to Lemma \ref{Lemma1}, the offloading time minimization problem of different users can be represented by 
    \begin{equation}
    \begin{aligned} \label{Problem_proof2}
    &\mathop {\min \quad \max }\limits_{\{\bm{\beta},\bm{p}\}}\left\{T_1^{off},T_2^{off}, \cdots,T_M^{off}\right\}.
    \end{aligned}
    \end{equation}
    Note that the minimum latency is  $T^*$ with the optimal solution $\{\bm{\beta}^*,\bm{p}^*\}$. This optimal solution is only obtained when $T_m^{off}=T_{m'}^{off},\forall m\neq m'$. Proof by contradiction can be exploited to prove this proposition.

    Assume that $U_{m'}$ decodes its signal firstly, and the optimal solution is obtained when $T_m^{off}>T_{m'}^{off}, m<m'$. Thus the minimum latency is $T^*$ with the optimal solution $\{\beta_m^*,\beta_{m'}^*,p_m^*,p_{m'}^*\}$. In this case, if we increase $p_m^*$ to $\hat{p}_m$, then $T_m^{off}$ will be decreased. Since $p_{m'}^*$ is fixed, and $T_{m'}^{off,*}$ will be increased. Therefore, there must exist $\hat{p}_m$ satisfying $T_m^{off}(p_m^*,p_{m'}^*)>\hat{T}_m^{off}(\hat{p}_m,p_{m'}^*)=\hat{T}_{m'}^{off}(\hat{p}_m,p_{m'}^*)>T_{m'}^{off}(p_m^*,p_{m'}^*)$. Therefore,  $\hat{T}_m^{off}(\hat{p}_m,p_{m'}^*)=\hat{T}_{m'}^{off}(\hat{p}_m,p_{m'}^*)$ should be the optimal time consumption since it has lower value than $T_m^{off}(\beta_m^*,p_m^*)$. This contradicts the assumption that  $\{\beta_m^*,p_m^*\}$ is the optimal solution to problem \eqref{Problem_proof2}.
    
    Assume that the optimal solution is obtained when $T_m^{off}<T_{m'}^{off}, m<m'$. Thus the minimum latency is $T^*$ with the optimal solution $\{\beta_m^*,\beta_{m'}^*,p_m^*,p_{m'}^*\}$. In this case, if we increase $p_m^*$ to $\hat{p}_m$, then $T_m^{off}$ will be decreased. Since $p_{m'}^*$ is fixed, and $T_{m'}^{off,*}$ will be increased. Therefore, there must exist $\hat{p}_m$ satisfying $T_m^{off}(p_m^*,p_{m'}^*)<\hat{T}_m^{off}(\hat{p}_m,p_{m'}^*)=\hat{T}_{m'}^{off}(\hat{p}_m,p_{m'}^*)<T_{m'}^{off}(p_m^*,p_{m'}^*)$. Therefore,  $\hat{T}_m^{off}(\hat{p}_m,p_{m'}^*)=\hat{T}_{m'}^{off}(\hat{p}_m,p_{m'}^*)$ should be the optimal time consumption since it has lower value than $T_{m'}^{off}(\beta_m^*,p_m^*)$. This contradicts the assumption that  $\{\beta_m^*,p_m^*\}$ is the optimal solution to problem \eqref{Problem_proof2}.
    
    Therefore, it can be concluded that the optimal solution to problem \eqref{Problem_proof2} can only be obtained when $T_m^{off}=T_{m'}^{loc}, \forall m\neq m'$.
\section{The Proof of Proposition \ref{quasi}} \label{QuasiProof}
The $\alpha$ sublevel sets of $\max \limits_{\{\beta_{m},p_{m}\}}  \left\{\tilde{T}_m^{off},T_m^{loc}, \forall m \right\}$ can be defined as 
\begin{equation}
S_{\alpha_T}=\{\{\beta_{m},p_{m}\}|\max{\left\{\tilde{T}_m^{off},T_m^{loc},\forall m \right\}}\leq \alpha_T\}.
\end{equation}
According to the definition of quasi-convex function \cite{BoydConv2004}, $\max{\left\{\tilde{T}_m^{off},T_m^{loc},\forall m\right\}}$ is quasi-convex if and only if its sublevel sets $S_{\alpha_T}$ is convex for any $\alpha_T$. Let's prove the $S_{\alpha_T}$ is convex set first. In our system, when $\alpha_T\leq0$, there are no solutions satisfying $\max{\left\{\tilde{T}_m^{off},T_m^{loc},\forall m \right\}}\leq\alpha_T$. When $\alpha_T> 0$, $S_{\alpha_T}$ can be rewritten as the following inequality:
\begin{subequations}
	\begin{align}
	&\frac{\sum\limits_{i=1}^m\beta_iL_i}{B\log_2(1+\sum\limits_{i=1}^m|h_i|^2p_i)}\leq \alpha_T, \forall m\label{T_Off}\\
	&\frac{(1-\beta_m)L_mC_m}{f_m^c}\leq\alpha_T, \forall m \label{T_1_c}
	\end{align}
\end{subequations}
It can be observed that \eqref{T_1_c} are linear equalities respect to $\beta_1$ and $\beta_2$, respectively. Therefore, \eqref{T_1_c} is convex sets for any $m$. To prove $S_{\alpha_T}$ is strictly quasi-convex, we need to prove that $\sum\limits_{i=1}^m\beta_{i}L_{i}-B\log_2(1+\sum\limits_{i=1}^m)$ is a strictly convex function. Since the first term $\sum\limits_{i=1}^m\beta_{i}L_{i}$ is liner function and the second term $-\log_2(1+\sum\limits_{i=1}^m)$ is convex function because $-\log(\cdot)$ is convex function and the term inside of $-\log$ is a liner function. Let's take two user case for example. The function in \eqref{T_Off}, $(\beta_1L_1+\beta_2L_2)/\alpha_T-B\log_2(1+|h_1|^2p_1^{off}+|h_2|^2p_2^{off})$, is a strictly convex function. To simplify the analysis, we use $a_1$, $a_2$, $b_1$ and $b_2$ to represent $L_1\alpha_T$, $L_2\alpha_T$, $|h_1|^2$ and $|h_2|^2$, respectively. We also use the variables $x_m$ and $y_m$ to respectively present $\beta_m$ and $p_m$ where $m=1, 2$. Thus the function in \eqref{T_Off} can be rewritten as $f(\bm{x},\bm{y})=(a_1x_1+a_2x_2)-B\log_2(1+b_1y_1+b_2y_2)$. To prove its convexity, we need to prove that its Hessian matrix is positive semi-definite matrix. The Hessian matrix of $f(\bm{x},\bm{y})$ is:
\begin{equation} 
\begin{aligned}
\mathbf{H}(f)=
\frac{1}{(1+b_1y_1+b_2y_2)^2\ln2}\left[
\begin{matrix}
0&0&0&0\\
0&0&0&0\\
0 &0&b_1^2 & b_1b_2 \\
0&0&b_1b_2 & b_2^2.
\end{matrix}
\right]
\end{aligned}
\end{equation}
According to the definition of semi-definite, we first define a non-zero column vector $\bm{v}=[v_1,v_2,v_3,v_4]$. To prove $\mathbf{H}$ is positive semi-definite matrix, we need to prove $\boldsymbol{v}^T\mathbf{H}\boldsymbol{v}\geq 0$. We have

\begin{equation} 
\begin{aligned}
\boldsymbol{v}^T\mathbf{H}(f)\boldsymbol{v}=
\frac{(b_1v_3+b_2v_4)^2}{(1+b_1y_1+b_2y_2)^2\ln2}.
\end{aligned}
\end{equation}
Since $v_3$, $v_4$, $b_1>0$ and $b_2>0$ have non-zero values, $\boldsymbol{v}^T\mathbf{H}(f)\boldsymbol{v}>0$. Therefore, Hessian matrix of $f$ is positive definite matrix. Hence \eqref{T_Off} is a convex set. Since all the sublevel sets are convex, the intersection of these convex sets is convex. Similarly, we could prove the constraint \eqref{eq:E_1_con} is convex set. Therefore, it can be concluded that problem \eqref{Prob:T_min_2} is a quasi-convex problem.

 \section{Proof of Proposition \ref{T_loc=T_off}}\label{T_m_loc=T_m_off}
 The task completion time minimization problem for each user, i.e., $U_m$, can be represented as
 \begin{equation}
 \begin{aligned} \label{Problem_proof}
 \mathop {\min \quad \max }\limits_{\{\beta_m,p_m\}}\left\{T_m^{off}, T_m^{loc}\right\}.
 \end{aligned}
 \end{equation}
 Note that the optimal solution to the above problem is $\{\beta_m^*,p_m^*\}$. This optimal solution is only obtained when $T_m^{off}=T_m^{loc}$. Proof by contradiction can be used to prove this proposition.
 
 Assume that the optimal solution is obtained when $T_m^{off}(\beta_m^*,p_m^*)> T_m^{loc}(\beta_m^*)$. Note that $E_m^{off}$ will increase and $E_m^{loc}$ will decrease when $\beta_m$ increases. Thus the optimal solution can be written as $T_m^{off}(\beta_m^*,p_m^*)$ satisfying the energy constraint and power constraint. In this case, there must exist $\hat{\beta}<\beta_m$ satisfying $T_m^{off}(\beta_m^*,p_m^*)>\hat{T}_m^{off}(\hat{\beta}_m,\hat{p}_m)=\hat{T}_m^{loc}(\hat{\beta}_m,\hat{p}_m)> T_m^{loc}(\beta_m^*)$. Therefore,  $\hat{T}_m^{off}(\hat{\beta}_m,\hat{p}_m)=\hat{T}_m^{loc}(\hat{\beta}_m,\hat{p}_m)$ should be the optimal time consumption since it has lower value than $T_m^{off}(\beta_m^*,p_m^*)$. This contradicts the assumption that $T^*$ is the minimum latency and $\{\beta_m^*,p_m^*\}$ is the optimal solution to problem \eqref{Problem_proof}.
 
 Similarly, assume that the optimal solution is obtained when $T_m^{off}(\beta_m^*,p_m^*)< T_m^{loc}(\beta_m^*)$. Note that $E_m^{off}$ will increase and $E_m^{loc}$ will decrease when $\beta_m$ increases. Thus the optimal solution can be written as $T_m^{off}(\beta_m^*,p_m^*)$ satisfying the energy constraint and power constraint. In this case, there must exist $\hat{\beta}>\beta_m$ satisfying $T_m^{off}(\beta_m^*,p_m^*)<\hat{T}_m^{off}(\hat{\beta}_m,\hat{p}_m)=\hat{T}_m^{loc}(\hat{\beta}_m,\hat{p}_m)< T_m^{loc}(\beta_m^*)$. Therefore,  $\hat{T}_m^{off}(\hat{\beta}_m,\hat{p}_m)=\hat{T}_m^{loc}(\hat{\beta}_m,\hat{p}_m)$ should be the optimal time consumption since it has lower value than $T_m^{off}(\beta_m^*,p_m^*)$. This contradicts the assumption that $\{\beta_m^*,p_m^*\}$ is the optimal solution to problem \eqref{Problem_proof}.
 
 Therefore, it can be concluded that the optimal solution to problem \eqref{Problem_proof} can only be obtained when $T_m^{off}=T_m^{loc}, \forall m$.

\section{Derivation of the Optimal Solution to Problem \eqref{Prob:p1,p2} }\label{OptimalDerivation2}
We rewrite problem \eqref{Prob:p1,p2} as follow
\begin{subequations}\label{Prob:p1}
	\begin{align}
	\mathop {\min }\limits_{ \{p_{1},p_{2}\}} \quad &\frac{a_1}{b_1+B\log_2(1+|h_1|^2p_1+|h_2|^2p_2)}\\
	\text{s.t.}  \quad	&0\leq p_{1}\leq P_{\max},0\leq p_{2}\leq P_{\max},  \label{p_range2}\\\nonumber
		&\kappa_1(1- \beta_1) L_1C_1 (f_1^{loc})^2\\
		&+\frac{\beta_1L_1+\beta_2L_2}{B\log_2(1+|h_1|^2p_1+|h_2|^2p_2)}p_1 \leq E_{\max}, \label{eq:E12}\\\nonumber
	&\kappa_2(1- \beta_2) L_2C_2 (f_2^{loc})^2 \\
	& +\frac{\beta_1L_1+\beta_2L_2}{B\log_2(1+|h_1|^2p_1+|h_2|^2p_2)}p_2 \leq E_{\max}\label{eq:E22}\\  \nonumber 
	\end{align}
\end{subequations}
where $a_1=L_1+L_2$ and $b_1=\frac{f_1^{loc}}{C_1}+\frac{f_2^{loc}}{C_2}$.
The Lagrangian function of problem \eqref{Prob:p1} can be written as:
\begin{equation}
\begin{aligned}
\mathcal{L}=&\frac{a_1}{b_1+B\log_2(1+|h_1|^2p_1+|h_2|^2p_2)}+\lambda_1(-p_1)\\
&+\lambda_2(p_1-P_{\max})+\lambda_3(-p_2)+\lambda_4(p_2-P_{\max})\\
&+\lambda_{5}\left(\kappa_1a_1(f_1^c)^3+a_1p_1-E_{\max}(b_1+R(p_1,p_2))\right)\\
&+\lambda_{6}\left(\kappa_1a_1(f_1^c)^3+a_1p_2-E_{\max}(b_1+R(p_1,p_2))\right)
\end{aligned}
\end{equation}
where $\lambda_i$ are Lagrangian multipliers corresponding to constraints in problem \eqref{Prob:p1}. Since it is convex problem and satisfies Slater's condition, thus KKT conditions are necessary and sufficient to obtain the optimal solution \cite{BoydConv2004}. To obtain the optimal solution,  KKT conditions (i.e., stationary condition, primal feasible condition, dual feasibility condition and Complementary slackness condition)  of problem \eqref{Prob:p1} can be written as follow:\\
Stationary condition:
\begin{equation}
	\begin{aligned}
	\frac{\partial\mathcal{L} }{\partial p_1}=&\frac{-a_1\frac{B|h_1|^2}{(1+|h_1|^2p_1^*+|h_2|^2p_2^*)\ln(2)}}{\left(b_1+B\log_2(1+|h_1|^2p_1^*+|h_2|^2p_2^*)\right)^2}-\lambda_1+\lambda_2\\
	&\label{partial_p1} 
	+\lambda_5\left(a_1-\frac{E_{\max}B|h_1|^2}{(1+|h_1|^2p_1^*+|h_2|^2p_2^*)\ln(2)}\right)=0
	\end{aligned}
\end{equation}
\begin{equation}
\begin{aligned}
\frac{\partial\mathcal{L} }{\partial p_2}=&\frac{-a_1\frac{B|h_2|^2}{(1+|h_1|^2p_1^*+|h_2|^2p_2^*)\ln(2)}}{\left(b_1+B\log_2(1+|h_1|^2p_1^*+|h_2|^2p_2^*)\right)^2}-\lambda_3+\lambda_4\label{partial_p2}\\ 
&+\lambda_6\left(a_1-\frac{E_{\max}B|h_2|^2}{(1+|h_1|^2p_1^*+|h_2|^2p_2^*)\ln(2)}\right)=0.
\end{aligned}
\end{equation}
Primal feasible condition:
\begin{subequations}
	\begin{align}
	&-p_1^*\leq 0,\ -p_2^*\leq 0\\
	&p_1^*-P_{\max}\leq 0,\ p_2^*-P_{\max}\leq 0\label{lambdaPmax1}\\ \nonumber
	&	\kappa_1a_1(f_1^{loc})^3+a_1p_1^*- \\ 
	& E_{\max}(b_1+B\log_2(1+|h_1|^2p_1^*+|h_2|^2p_2^*))\leq 0\label{primalE1} \\\nonumber
    &	\kappa_2a_1(f_1^{loc})^3+a_1p_2^* \\ 
    &- E_{\max}(b_1+B\log_2(1+|h_1|^2p_1^*+|h_2|^2p_2^*))\leq 0\label{primalE2}. \\ \nonumber
	\end{align}
\end{subequations}
Dual feasibility condition: $\lambda_i \geq 0, i={1, \cdots, 6}$.
Stationary condition:
\begin{subequations}
	\begin{align}
	&\lambda_1p_1^*= 0,\ \lambda_3p_2^*= 0\label{lambdap1}\\
	&\lambda_{2}\left(p_1^*-P_{\max}\right)=0,\ \lambda_{4}\left(p_2^*-P_{\max}\right)=0 \label{lambdaPmax}\\ \nonumber
	&\lambda_{5}\left[\kappa_1a_1(f_1^{loc})^3+a_1p_1^*\right.\\ &\left. -E_{\max}\left(b_1+B\log_2\left(1+|h_1|^2p_1^*+|h_2|^2p_2^*\right)\right) \right]= 0 \label{lambdaE1}\\  
		&\lambda_{6}\left[\kappa_2a_1(f_1^{loc})^3+a_1p_2^*\right.\\
		&\left.  - E_{\max}\left(b_1+B\log_2\left(1+|h_1|^2p_1^*+|h_2|^2p_2^*\right)\right) \right]= 0. \label{lambdaE2}\\ \nonumber 
	\end{align}
\end{subequations}

In our solution, we have $p_1^*>0$ and $p_2^*>0$. To satisfy \eqref{lambdap1}, we can obtain $\lambda_{1}=0$ and $\lambda_{3}=0$. To satisfy \eqref{partial_p1}, when $a_1\geq \frac{E_{\max}B|h_2|^2}{(1+|h_1|^2p_1^*+|h_2|^2p_2^*)\ln(2)}$, we have at lease one of $\lambda_2$ and $\lambda_5$ is larger than zero. There are three cases to satisfy this condition: 1. $\lambda_2>0, \lambda_5=0$; 2. $\lambda_2=0, \lambda_5>0$; 3. $\lambda_2>0, \lambda_5>0$. 
Let's first investigate these three cases. 

When $\lambda_2>0,\lambda_5=0$, to satisfy \eqref{lambdaPmax}, we have $p_1^*=P_{\max}$
Based on $p_1^*$, to obtain $p_2^*$, we have the following calculation steps.
To satisfy \eqref{partial_p2}, due to the negativity of the first term, we have at least one of $\lambda_4$ and $\lambda_{6}$ is larger than zero based on $a_1>\frac{E_{\max}B|h_2|^2}{(1+|h_1|^2p_1^*+|h_2|^2p_2^*)\ln(2)}$. When $a_1\leq \frac{E_{\max}B|h_2|^2}{(1+|h_1|^2p_1^*+|h_2|^2p_2^*)\ln(2)}$, we must have $\lambda_{4}>0$. To be concluded, we have three cases to obtain $p_2^*$: 1. $\lambda_4>0, \lambda_6=0$; 2. $\lambda_4=0, \lambda_6>0$ based on $a_1\geq\frac{E_{\max}B|h_2|^2}{(1+|h_1|^2p_1^*+|h_2|^2p_2^*)\ln(2)}$; 3. $\lambda_4>0, \lambda_6>0$.

\emph{\underline{Case 1}}: $\lambda_2>0,\lambda_5=0$, $\lambda_4>0, \lambda_6=0$, to satisfy \eqref{lambdaPmax}, we have 
\begin{equation}
p_2^*=P_{\max}. \label{optimal p2=pmax}
\end{equation}
In this case, to satisfy the feasible conditions \eqref{primalE1} and \eqref{primalE2}, we must have
\begin{equation}
p_1^*\leq P_{1,w}(p_2^*=P_{\max})
\end{equation} 
where $P_{1,w}(p_2^*)$ is defined in \eqref{p1range_p1}.
\begin{equation}
p_2^*\leq P_{2,w}(p_1^*=P_{\max})
\end{equation}
where $P_{2,w}(p_2^*)$ is defined in \eqref{p2range_p1}.
Therefore, we have the optimal solution \eqref{Case1Opt}.

\emph{\underline{Case 2}}: $\lambda_2>0,\lambda_5=0$, $\lambda_4=0$ and $\lambda_6>0$, to satisfy \eqref{lambdaE2}, we have 
\begin{equation}
\begin{aligned}
&\kappa_2a_1(f_2^{loc})^3+a_1p_2^*  \\&- E_{\max}\left(b_1+B\log_2\left(1+|h_1|^2p_1^*+|h_2|^2p_2^*\right)\right) = 0
\end{aligned}
\end{equation}
Thus we have
\begin{equation}
p_2^*=P_{2,w}(p_1^*=P_{\max}). \label{optimal p2=W}
\end{equation}
In this case, we need to satisfy primal feasible condition as $ P_{2,w}(p_1^*=P_{\max})\leq P_{\max}$
Therefore, we can have the optimal solution as \eqref{Case1Opt}.
 
When $\lambda_2>0,\lambda_5=0$, $\lambda_4>0$ and $\lambda_6>0$, to satisfy \eqref{lambdaPmax} and \eqref{lambdaE2}, we have \eqref{optimal p2=pmax}=\eqref{optimal p2=W}, which leads to
\begin{equation}
P_{\max}=P_{2,w}(p_1^*=P_{\max}). 
\end{equation}	
This case can be included into the above cases, thus this case can be ignored.
	
When $\lambda_2=0, \lambda_5>0$, to satisfy \eqref{lambdaE1}, we have 
\begin{equation}
\begin{aligned}
      &\kappa_1a_1(f_1^{loc})^3+a_1p_1^*\\& - E_{\max}\left(b_1+B\log_2\left(1+|h_1|^2p_1^*+|h_2|^2p_2^*\right)\right) = 0. 
\end{aligned}
\end{equation}
Thus we have
\begin{equation}
p_1^*=P_{1,w}(p_2^*). \label{optimal p1=W}
\end{equation}
To obtain $p_2^*$, we also need to discus three cases: i). $\lambda_4>0, \lambda_6=0$; ii). $\lambda_4=0, \lambda_6>0$ based on $a_1\geq\frac{E_{\max}B|h_2|^2}{(1+|h_1|^2p_1^*+|h_2|^2p_2^*)\ln(2)}$; iii). $\lambda_4>0, \lambda_6>0$. 
similarly, Case iii) can be included within case i) or ii).

\emph{\underline{Case 3}}: $\lambda_2=0, \lambda_5>0$, $\lambda_4>0$ and $\lambda_6=0$, to satisfy \eqref{lambdaPmax}, we have 
\begin{equation}
p_2^*=P_{\max}. \label{optimal p2=pmax}
\end{equation}
In this case, to satisfy the feasible conditions \eqref{primalE1} and \eqref{primalE2}, we must have \eqref{Case3Opt}.

\emph{\underline{Case 4}}: $\lambda_2=0, \lambda_5>0$, $\lambda_4=0$ and $\lambda_6>0$, to satisfy \eqref{lambdaE1} and \eqref{lambdaE2}, we have \eqref{Case4}.
Therefore, we can have the optimal solution \eqref{Case4}. 

Above all, the optimal solution of problem \eqref{Prob:p1,p2} can be concluded as four cases. 


\bibliographystyle{IEEEtran}
\vspace{-0.5cm}

\end{document}